\documentclass[reqno, eucal]{amsart}

\usepackage[mathscr]{eucal} 
\usepackage[dvips]{graphicx}
\usepackage[dvips]{color}
\usepackage{amsmath,amsfonts,amssymb,amsthm,amscd}
\usepackage{epic}
\usepackage{eepic}
\usepackage{longtable}
\usepackage{array}
\usepackage{here}
\numberwithin{equation}{section}

\newtheorem{theorem}{Theorem}[section]
\newtheorem{lemma}[theorem]{Lemma}

\newtheorem{proposition}[theorem]{Proposition}
\newtheorem{corollary}[theorem]{Corollary}

\theoremstyle{definition}

\newtheorem{example}[theorem]{Example}
\newtheorem{remark}[theorem]{Remark}

\newcommand{\Z}{{\mathbb Z}}
\newcommand{\R}{{\mathbb R}}
\newcommand{\C}{{\mathbb C}}

\begin{document}

\title[Inhomogeneous TAZRP]{Inhomogeneous generalization of\\
multispecies totally asymmetric zero range process}

\author{Atsuo Kuniba}
\email{atsuo@gokutan.c.u-tokyo.ac.jp}
\address{Institute of Physics, University of Tokyo, Komaba, Tokyo 153-8902, Japan}

\author{Shouya Maruyama}
\email{maruyama@gokutan.c.u-tokyo.ac.jp}
\address{Institute of Physics, University of Tokyo, Komaba, Tokyo 153-8902, Japan}

\author{Masato Okado}
\email{okado@sci.osaka-cu.ac.jp}
\address{Department of Mathematics, Osaka City University, 
3-3-138, Sugimoto, Sumiyoshi-ku, Osaka, 558-8585, Japan}


\maketitle

\vspace{0.5cm}
\begin{center}{\bf Abstract}
\end{center}
The $n$-species totally asymmetric zero range process ($n$-TAZRP) 
on one-dimensional periodic chain studied recently by the authors
is a continuous time Markov process where arbitrary number of 
particles can occupy the same sites and hop to the adjacent sites 
only in one direction with a priority constraint according to their species.
In this paper we introduce an $n$-parameter generalization of the 
$n$-TAZRP having inhomogeneous transition rate.
The steady state probability is obtained in 
a matrix product form and also by an algorithm related to combinatorial $R$.

\vspace{0.5cm}

\section{Introduction}\label{sec:intro}

Zero range processes on lattice 
are stochastic particle systems modeling various nonlinear dynamics
in biology, chemistry, physics, networks and so on  
where the jump rates are determined by those sharing the same departure site 
\cite{S, A, KL, GSS, EH}.
In \cite{KMO3} a new multispecies totally asymmetric zero range process, 
called $n$-TAZRP,
on one-dimensional periodic chain was proposed.
There are $n$-species of particles that can occupy sites of the chain without 
an exclusion rule, and they hop to the adjacent sites only in one direction with 
the constraint that larger species ones have the priority to do so\footnote{
The enumeration of particles is reversed from \cite{KMO3} where the 
smaller species ones had the priority.}.
It is the first example of $n$-species models of zero range interaction
that allows a matrix product formula for arbitrary $n$  
and possesses a rich integrable structure \cite{KMO3,KMO4} 
related to crystals of quantum groups \cite{Ka1}
and the tetrahedron equation \cite{Zam80}.

The proposal of the $n$-TAZRP was 
preceded by the discovery \cite{KMO1,KMO2} of the 
quite parallel features in 
the $n$-species totally asymmetric simple exclusion process ($n$-TASEP).
We refer to \cite{PEM,AKSS,TW} for results on 
more general $n$-ASEP that have been established for arbitrary $n$.
In our approach,  the $n$-TAZRP and the $n$-TASEP 
on the length $L$ periodic chain turn out to be the sister models
associated with crystals \cite{NY} of the symmetric 
and the anti-symmetric tensor representation of the quantum affine algebra 
$U_q(\widehat{sl}_L)$ \cite{D86,J}, respectively.
Their matrix product formulae \cite{KMO1,KMO3} are traced back to
the factorization of quantum $R$ matrices \cite{KOS}
based on the $\mathscr{R}$ and the 
$\mathscr{L}$-operators obeying the tetrahedron equation, respectively.
The integer $n$ plays the role of a system size of the associated three-dimensional
lattice models.

The purpose of this paper is to introduce and study an $n$-parameter generalization 
of the $n$-TAZRP having the inhomogeneous transition rate $w_1, \ldots, w_n$.
We call it $n$-species inhomogeneous totally asymmetric 
zero range process ($n$-iTAZRP).
Under an appropriate normalization, the steady state probabilities of the $n$-iTAZRP 
become homogeneous
polynomials of $w_1, \ldots, w_n$ with nonnegative integer coefficients.
We establish a matrix product formula and a combinatorial algorithm to calculate them
as the generating functions of the set $B({\bf m})$ (\ref{mkr}) 
with a certain {\em weight}  $W$ (\ref{hzks}).
The procedure is an iTAZRP analogue of the queueing process approach 
to TASEP \cite{FM}, and provides 
$n$-kinds of statistics on the crystal of the symmetric tensor representation.

A similar inhomogeneous generalization has been known for the $n$-TASEP  
with $n=2$ \cite{AL} and general $n$ \cite{AM}. 
See also \cite{LW}.
In fact our derivation of the main results in this paper is based on 
the so-called generalized hat relation having the same form as \cite{AM}.
It is an open problem to incorporate it in the framework of \cite{KMO3,KMO4}.

The layout of the paper is as follows.
In Section \ref{sec:taz} the $n$-iTAZRP and its steady states are defined.
In Section \ref{sec:mpf} a matrix product formula for the steady state probability
is presented.
In Section \ref{sec:mp} a combinatorial algorithm for calculating the 
steady state probability is given.
Section \ref{sec:d} is a discussion.
Our main results are Theorem \ref{th:P} and Theorem \ref{th:we}, and they are 
consequences of the generalized hat relation in Proposition \ref{pr:hat}.

Throughout the paper we set
$[i,j]=\{k \in \Z\mid i \le k \le j\}$, and 
use the characteristic function $\theta$ defined by 
$\theta(\text{true}) = 1, \theta(\text{false}) = 0$ and the symbol
$\delta^{\alpha_1,\ldots, \alpha_m}_{\beta_1,\ldots, \beta_m} = 
\prod_{j=1}^m\theta(\alpha_j=\beta_j)$.

\section{$n$-$\mathrm{i}$TAZRP}\label{sec:taz}
\subsection{Definition of $n$-iTAZRP}
Consider a periodic one-dimensional chain $\Z_L$ with $L$ sites.
Each site $i \in \Z_L$ is assigned with a local state 
$\sigma_i=(\sigma_i^1,\ldots, \sigma^n_i)\in (\Z_{\ge 0})^n$ which is 
interpreted as an assembly of $n$ species of particles as
\begin{equation}\label{cie}
\begin{picture}(260,25)(-80,0)

\put(-1,5){$\overbrace{1 \ldots 1}^{\sigma^1_i} \,
\overbrace{2 \ldots 2}^{\sigma^2_i} \,
\ldots \,
\overbrace{n \ldots n}^{\sigma^n_i} $}
\put(7,0){\put(-16,1){\line(1,0){115}}\put(-16,1){\line(0,1){12}}
\put(99,1){\line(0,1){12}}}

\end{picture}
\end{equation}
The ordering of particles within a site does not matter.
A local state $\alpha$ is specified uniquely either by 
{\em multiplicity representation}
$\alpha=(\alpha^1,\ldots, \alpha^n) \in (\Z_{\ge 0})^n$ as above or 
{\em multiset representation}
$\alpha=(\alpha_1,\ldots, \alpha_r) \in [1,n]^r$ 
with $1 \le \alpha_1 \le \cdots \le \alpha_r\le n$.
For example the 4-iTAZRP local state
$(3,0,2,1)$ in the former is $(1,1,1,3,3,4)$ in the latter.
In general they are related by
$\alpha^a = \#\{j \in [1,r]\mid \alpha_j = a\}$ and 
$r = |\alpha|:=\alpha^1 + \cdots + \alpha^n$.

Let $({\alpha}, {\beta})$ 
and $({\gamma},{\delta})$ be pairs of local states.
Let $(\beta_1, \ldots, \beta_r)$ 
be the multiset representation of the $\beta$, hence 
$1 \le \beta_1 \le \cdots \le \beta_r \le n$.
For the two pairs we define $>$ by
\begin{align}\label{kyk}
({\alpha}, {\beta}) > ({\gamma},{\delta})
\overset{\text{def}}{\Longleftrightarrow}
\gamma = \alpha \cup \{\beta_{k}, \beta_{k+1}, \ldots, \beta_r\},\;
\delta=(\beta_1,\beta_2, \ldots, \beta_{k-1})\;\;
\text{for some}\; k \in [1,r],
\end{align}
where $\alpha \cup \{\beta_{k}, \beta_{k+1}, \ldots, \beta_r\}$ 
is a union as a multiset.
For instance in the multiset representation we have\footnote{
Here and in what follows, 
a multiset (set accounting for multiplicity of elements), say  $\{1,1,3,5,6\}$,
is abbreviated to $11356$, which does cause a confusion since
all the examples in this paper shall be concerned with the case $n\le 9$.}
\begin{equation}\label{ykw}
\begin{split}
(124,2335)&> (1245,233),  (12345,23), (123345,2), (1223345,\emptyset),\\
(235,12446) &> (2356,1244), (23456,124), (234456,12), 
(2234456,1), (12234456, \emptyset),\\
(\emptyset,255) & > (5,25), (55,2), (255,\emptyset),\\
(3446,\emptyset) &> \text{none}.
\end{split}
\end{equation} 

Let $w_1, \ldots, w_n \in \R_{> 0}$ be parameters.
By {\em inhomogeneous 
$n$-species totally asymmetric zero range process} ($n$-iTAZRP)
we mean a stochastic process on $\Z_L$ 
in which neighboring pairs of local states
$(\sigma_i, \sigma_{i+1})=(\alpha, \beta)$ 
change into $(\sigma'_i, \sigma'_{i+1})=(\gamma, \delta)$ such that 
$({\alpha}, {\beta}) > ({\gamma},{\delta})$  
with the transition rate $w_{\beta_k}$ in the situation (\ref{kyk})
depicted as 
\begin{equation*}
\begin{picture}(120,50)(27,-13)


\put(-85,5){${\alpha}_1\, \ldots \,{\alpha}_s$}
\put(-13,5){${\beta}_1 \ldots \ldots .\,. \, {\beta}_r$}
\put(-110,0){

\put(65,26){${\beta}_k \ldots {\beta}_r$}
\put(41.5,22){\vector(0,-1){9}}\put(41.5,22){\line(1,0){83}}
\put(124.5,22){\line(0,-1){8}}

\put(0,0){\line(1,0){166}}\put(0,0){\line(0,1){12}}
\put(83,0){\line(0,1){12}}\put(166,0){\line(0,1){12}}

\put(28,-12){$\sigma_i=\alpha$}\put(105,-11){$\sigma_{i+1}=\beta$}

}

\put(53,27){rate $w_{\beta_k}$}
\put(63,6){\vector(1,0){15}}

\put(-10,0){
\put(100,5){${\alpha}_1 \ldots {\alpha}_s \,
{\beta}_k \ldots {\beta}_r$}
\put(195,5){${\beta}_{1} \,\ldots \,{\beta}_{k-1}$}
\put(95,0){\put(0,0){\line(1,0){166}}\put(0,0){\line(0,1){12}}
\put(83,0){\line(0,1){12}}\put(166,0){\line(0,1){12}}
\put(28,-12){$\sigma'_i=\gamma$}\put(104,-11){$\sigma'_{i+1}=\delta$}
}}
\end{picture}
\end{equation*}
Here $1 \le {\beta}_1 \le \cdots \le 
{\beta}_r \le n, \, i \in \Z_L $ and 
$k \in [1,r]$ is arbitrary. 
Note that $\beta_k = \mathrm{min}(\gamma\!\setminus \!\alpha)$ 
in the multiset representation.
Thus the rate can take arbitrary values 
depending on the {\em minimum} of the species of the 
particles that have hopped to the left. 
As an example, the first line in (\ref{ykw}) implies that 
the following local transitions take place with the respective rate:
\begin{align*}
\begin{picture}(500,175)(-100,-123)
\put(53,40){local transitions} \put(210,40){rate}
\put(12,3.5){124} \put(50,3.5){2335}
\put(97,0){\put(12,3.5){1245} \put(55.5,3.5){233}}
\put(216,3.5){$w_5$}
\put(38,25){5}
\put(20,22){\vector(0,-1){8}}\put(20,22){\line(1,0){40}}
\put(60,22){\line(0,-1){7}}
\put(0,0){\line(1,0){80}}
\put(0,0){\line(0,1){10}}
\put(40,0){\line(0,1){10}}
\put(80,0){\line(0,1){10}}
\put(85, 6){\vector(1,0){10}}
\put(100,0){
\put(0,0){\line(1,0){80}}
\put(0,0){\line(0,1){10}}
\put(40,0){\line(0,1){10}}
\put(80,0){\line(0,1){10}}}
%
\put(0,-40){
\put(12,3.5){124} \put(50,3.5){2335}
\put(97,0){\put(11,3.5){12345} \put(58,3.5){23}}
\put(216,3.5){$w_3$}
\put(36,25){35}
\put(20,22){\vector(0,-1){8}}\put(20,22){\line(1,0){40}}
\put(60,22){\line(0,-1){7}}
\put(0,0){\line(1,0){80}}
\put(0,0){\line(0,1){10}}
\put(40,0){\line(0,1){10}}
\put(80,0){\line(0,1){10}}
\put(85, 6){\vector(1,0){10}}
\put(100,0){
\put(0,0){\line(1,0){80}}
\put(0,0){\line(0,1){10}}
\put(40,0){\line(0,1){10}}
\put(80,0){\line(0,1){10}}}
}
\put(0,-80){
\put(12,3.5){124} \put(50,3.5){2335}
\put(97,0){\put(8,3.5){123345} \put(60,3.5){2}}
\put(216,3.5){$w_3$}
\put(34,25){335}
\put(20,22){\vector(0,-1){8}}\put(20,22){\line(1,0){40}}
\put(60,22){\line(0,-1){7}}
\put(0,0){\line(1,0){80}}
\put(0,0){\line(0,1){10}}
\put(40,0){\line(0,1){10}}
\put(80,0){\line(0,1){10}}
\put(85, 6){\vector(1,0){10}}
\put(100,0){
\put(0,0){\line(1,0){80}}
\put(0,0){\line(0,1){10}}
\put(40,0){\line(0,1){10}}
\put(80,0){\line(0,1){10}}}
}
\put(0,-120){
\put(12,3.5){124} \put(50,3.5){2335}
\put(97,0){\put(6,3.5){1223345} }
\put(216,3.5){$w_2$}
\put(30.5,25){2335}
\put(20,22){\vector(0,-1){8}}\put(20,22){\line(1,0){40}}
\put(60,22){\line(0,-1){7}}
\put(0,0){\line(1,0){80}}
\put(0,0){\line(0,1){10}}
\put(40,0){\line(0,1){10}}
\put(80,0){\line(0,1){10}}
\put(85, 6){\vector(1,0){10}}
\put(100,0){
\put(0,0){\line(1,0){80}}
\put(0,0){\line(0,1){10}}
\put(40,0){\line(0,1){10}}
\put(80,0){\line(0,1){10}}}
}
\end{picture}
\end{align*}

This dynamics is {\em totally asymmetric} in that 
particles can hop only to the left adjacent site.
Their interaction is of {\em zero range} in that 
the hopping priority for larger species particles is respected 
only among those occupying the same departure site and 
no constraint is imposed on the status of the destination site 
nor the number of particles that hop at a transition.  
It is {\em inhomogeneous} for $n\ge 2$
in that the transition rate depends on those hopping particles.
A pair $(\alpha, \beta)$ of adjacent local states has $|\beta|$ 
possibilities to change into.
(The symbol $|\beta|$ has been defined after (\ref{cie}).)
This model was first introduced in \cite{KMO3}  and further studied in \cite{KMO4} 
for the homogeneous case $w_1=\cdots = w_n = 1$ in the opposite convention that 
smaller species particles have the priority to move.

The $n$-iTAZRP dynamics obviously preserves 
the number of particles of each species.
Thus the problem splits into {\em sectors} labeled with 
{\em multiplicity} ${\bf m}=(m_1,\ldots, m_n) \in (\Z_{\ge 0})^n$ of 
the species of particles:
\begin{align}\label{kgc}
S({\bf m}) = 
\{{\boldsymbol \sigma}=(\sigma_1,\ldots, \sigma_L)\mid
\sigma_i = (\sigma^{1}_i, \ldots, \sigma^{n}_i)  \in (\Z_{\ge 0})^n,\;
\sum_{i=1}^L \sigma_i= {\bf m}\}.
\end{align}
A sector $S({\bf m})$ such that $m_a \ge 1$ for all $a \in [1,n]$
is called {\em basic}.
Non-basic sectors are equivalent to a basic sector for $n'$-iTAZRP with some 
$n'<n$ by a suitable relabeling of species and $w_k$'s.
Henceforth we shall exclusively deal with basic sectors in this paper.

A local state $\sigma_i$ in (\ref{kgc}) 
can take $N=\prod_{a=1}^n(m_a+1)$ possibilities in view of (\ref{cie}).
Let $\{|{\boldsymbol \sigma} \rangle = |\sigma_1,\ldots, \sigma_L\rangle\}$ 
be a basis of
$(\C^N)^{\otimes L}$.
Denoting by ${\mathbb P}(\sigma_1,\ldots, \sigma_L; t)$ the probability of finding 
the system in the configuration 
${\boldsymbol \sigma}=(\sigma_1,\ldots, \sigma_L)$ at time $t$,  we set 
\begin{align*}
|P(t)\rangle
= \sum_{{\boldsymbol \sigma} \in S({\bf m})}
{\mathbb P}(\sigma_1,\ldots, \sigma_L; t)|\sigma_1,\ldots, \sigma_L\rangle.
\end{align*}
This actually belongs to a
subspace of $(\C^N)^{\otimes L}$ of dimension 
$\# S({\bf m}) = \prod_{a=1}^n\binom{L+m_a-1}{m_a}$ which is in general much 
smaller than $N^L$ reflecting the constraint in (\ref{kgc}).

Our $n$-iTAZRP is a continuous-time Markov process 
governed by the master equation
\begin{align*}
\frac{d}{dt}|P(t)\rangle
= H_{\mathrm{iTAZRP}} |P(t)\rangle,
\end{align*}
where the Markov matrix has the form
\begin{align}\label{hrk2}
H_{\mathrm{iTAZRP}}  = \sum_{i \in \Z_L} h_{i,i+1},\qquad
h |\alpha, \beta\rangle =\sum_{\gamma,\delta}h^{\gamma,\delta}_{\alpha,\beta}
|\gamma,\delta\rangle.
\end{align}
Here $h_{i,i+1}$ is the local Markov matrix that  
acts as $h$ on the $i$-th and the $(i+1)$-th components non-trivially and 
as the identity elsewhere.
If the transition rate of the adjacent pair of local states
$(\alpha, \beta) \rightarrow (\gamma,\delta)$ is denoted by 
$w(\alpha\beta \rightarrow \gamma\delta)$,
the matrix element of the $h$ is given by 
$h^{\gamma,\delta}_{\alpha,\beta}
= w(\alpha\beta \rightarrow \gamma\delta)
-\theta\bigl((\alpha,\beta)=(\gamma,\delta)\bigr)
\sum_{\gamma',\delta'}w(\alpha\beta \rightarrow \gamma'\delta')$.
Our $n$-iTAZRP corresponds to the choice 
$w(\alpha\beta \rightarrow \gamma\delta) = 
\theta\bigl((\alpha,\beta)>(\gamma,\delta)\bigr)w_{\mathrm{min}(\gamma\setminus\alpha)}$, 
therefore the general formula,
which is independent of $w(\alpha\beta \rightarrow \alpha\beta)$, gives
\begin{align}
h^{\gamma,\delta}_{\alpha,\beta}&=
\begin{cases}
w_{\mathrm{min}(\gamma\setminus\alpha)} & \text{if }\;(\alpha, \beta) > (\gamma,\delta),\\
-g(\beta) & \text{if } \;(\alpha, \beta)  = (\gamma,\delta),\\
0 & \text{otherwise}.
\end{cases}\label{hdef}\\
\qquad g(\beta) &= w_1\beta^1+\cdots + w_n\beta^n
\;\;\text{for}\; \;\beta=(\beta^1,\ldots, \beta^n)\;\;
\text{in multiplicity representation}.
\label{gdef}
\end{align}

\subsection{Steady state}
Given a system size $L$ and a sector $S({\bf m})$
there is a unique vector 
\begin{align*}
|\mathscr{P}_L({\bf m})\rangle = \sum_{{\boldsymbol \sigma} \in S({\bf m})}
\mathbb{P}({\boldsymbol \sigma}) |{\boldsymbol \sigma} \rangle
\end{align*} 
up to a normalization, called the {\em steady state},
which satisfies $H_{\mathrm{iTAZRP}}  
|\mathscr{P}_L({\bf m})\rangle=0$ hence is 
time-independent.
From (\ref{hrk2}) 
we see that $\mathbb{P}({\boldsymbol \sigma})$ is the 
(unique up to overall) solution 
to the linear equation
\begin{align}\label{htn}
\sum_{i \in \Z_L} \sum_{\sigma'_i, \sigma'_{i+1}}
h^{\sigma_i, \sigma_{i+1}}_{\sigma'_i, \sigma'_{i+1}}
\mathbb{P}(\sigma_1,\ldots, \sigma_{i-1},
\sigma'_i, \sigma'_{i+1},
\sigma_{i+2},\ldots, \sigma_L) = 0
\quad \forall (\sigma_1,\ldots, \sigma_L) \in S({\bf m}),
\end{align}
where $\sigma'_i, \sigma'_{i+1}$ range over $(\Z_{\ge 0})^n$.
In what follows we shall take ${\mathbb P}({\boldsymbol \sigma})$'s 
to be the degree $(n-1)(L-1)$ 
homogeneous {\em polynomials} 
in $\Z_{\ge 0}[w_1, \ldots, w_n]$\footnote{That such a choice is possible 
will be shown by Theorem \ref{th:P} and Lemma \ref{le:iii}.} such that 
\begin{align*}
\sum_{{\boldsymbol \sigma} \in S({\bf m})}
\mathbb{P}({\boldsymbol \sigma})|_{w_1=\cdots = w_n=1} = \prod_{a=1}^n
\binom{L-1+\ell_a}{\ell_a},
\end{align*}
where $\ell_a$ is specified in (\ref{mkr:akci}) and 
the right hand side is $\# B({\bf m})$ (\ref{mkr}).
This is natural in view of Theorem \ref{th:we} and agrees with \cite{KMO3}.
The unnormalized 
${\mathbb P}({\boldsymbol \sigma})$ will be called the steady state probability
by abusing the terminology. 

For $n=1$, all the local transitions have the common rate $w_1$ 
hence there is actually no inhomogeneity.
The steady state is trivial 
under the present periodic boundary condition in that  
all the configurations ${\boldsymbol \sigma}=(\sigma_1,\ldots, \sigma_L) \in 
(\Z_{\ge 0})^L$ in a given sector are realized with the equal probability.
This can be seen by noting that the numbers of configurations 
jumping into and out the ${\boldsymbol \sigma}$
are both equal to 
$\sigma_1+\cdots + \sigma_L$.
The relation (\ref{htn}) reduces to 
$\sum_{i \in \Z_L} \sum_{\alpha, \beta}
h^{\sigma_i, \sigma_{i+1}}_{\alpha,\beta}=0$ 
which is valid because of 
$\sum_{\alpha, \beta}h^{\gamma, \delta}_{\alpha,\beta}
=(\gamma-\delta)w_1$.
The steady states for the $n$-iTAZRP with $n \ge 2$ are nontrivial.

\begin{example}\label{ex:LL}
We present the steady state in small sectors of
2-iTAZRP and 3-iTAZRP in the form
\begin{align*}
|\mathscr{P}_L({\bf m})\rangle = |\xi_L({\bf m})\rangle
+ C|\xi_L({\bf m})\rangle + \cdots + 
C^{L-1} |\xi_L({\bf m})\rangle
\end{align*}
respecting the symmetry 
$H_{\mathrm{iTAZRP}} C=CH_{\mathrm{iTAZRP}} $ 
under the $\Z_L$ cyclic shift
 $C: |\sigma_1, \sigma_2,\ldots, \sigma_L\rangle \mapsto 
 |\sigma_L, \sigma_1, \ldots, \sigma_{L-1}\rangle$.
The choice of the vector $|\xi_L({\bf m})\rangle$ is not unique.
We employ multiset representation like 
$|\emptyset, 3, 122\rangle$, which would have looked as 
$|000,001,120\rangle$ in the multiplicity representation for the 3-iTAZRP.

For the 2-iTAZRP  one has
\begin{align*}
|\xi_2(1,1)\rangle&= (w_1+w_2)|\emptyset, 12\rangle + w_2| 1, 2\rangle,\\
|\xi_3(1,1)\rangle& = (w_1^2+w_1w_2+w_2^2)|\emptyset, \emptyset, 12\rangle +
w_2^2|\emptyset, 1, 2\rangle +
w_2(w_1+w_2)|\emptyset, 2, 1\rangle,\\
|\xi_4(1,1)\rangle& = (w_1+w_2)(w_1^2+w_2^2)|\emptyset, \emptyset, \emptyset, 12\rangle +
w_2^3|\emptyset, \emptyset, 1,2 \rangle +
w_2^2 (w_1 + w_2)|\emptyset, 1,\emptyset, 2 \rangle \\
&+w_2 (w_1^2 + w_1w_2 + w_2^2)|\emptyset, \emptyset, 2,1 \rangle,\\
|\xi_2(2,1)\rangle& =
(2 w_1 + w_2)|\emptyset, 112\rangle +(w_1 + w_2) |1,12\rangle + w_2|2, 11\rangle,\\
|\xi_3(2,1)\rangle &=
(3 w_1^2 + 2 w_1 w_2 + w_2^2)|\emptyset, \emptyset, 112\rangle +
(w_1^2 + w_1 w_2 + w_2^2)|\emptyset, 1,12\rangle
+ w_2 (2 w_1 + w_2)|\emptyset, 2, 11\rangle \\
&+
w_2^2|\emptyset, 11, 2\rangle +
(2 w_1^2 + 2 w_1 w_2 + w_2^2)|\emptyset, 12,1\rangle +
w_2 (w_1 + w_2)|1,1,2\rangle,\\
|\xi_2(1,2)\rangle &=
(w_1 + w_2)|\emptyset,122\rangle + w_2|1,22\rangle + w_2 |2,12\rangle,\\
|\xi_3(1,2)\rangle &=
(w_1^2 + w_1 w_2 + w_2^2)|\emptyset, \emptyset, 122\rangle +
w_2^2|\emptyset, 1, 22\rangle +
w_2 (w_1 + w_2)|\emptyset, 2, 12\rangle +
w_2^2|\emptyset, 12, 2\rangle \\
&+
w_2 (w_1 + w_2)|\emptyset, 22, 1\rangle +
w_2^2|1,2,2\rangle,\\
|\xi_2(3,1)\rangle&=
(3 w_1 + w_2) |\emptyset, 1112\rangle 
+(2 w_1 + w_2) |{1}, 112\rangle 
+w_2 |{2}, 111\rangle + 
 (w_1 + w_2)|11, 12\rangle,\\
|\xi_3(3,1)\rangle&=
 (6 w_1^2 + 3 w_1 w_2 + w_2^2)|\emptyset, \emptyset, 1112\rangle 
+ (3 w_1^2 + 2 w_1 w_2 + w_2^2) |\emptyset, {1}, 112\rangle \\
&+ w_2 (3 w_1 + w_2)|\emptyset, {2}, 111\rangle
+ (w_1^2 + w_1 w_2 + w_2^2) |\emptyset, 11, 12\rangle\\
& + (3 w_1^2 + 3 w_1 w_2 + w_2^2)|\emptyset, 12, 11\rangle
+ w_2^2 |\emptyset, 111, {2}\rangle \\
&+ (5 w_1^2 + 3 w_1 w_2 + w_2^2)|\emptyset, 112, {1}\rangle 
 +(2 w_1^2 + 2 w_1 w_2 + w_2^2) |{1}, {1}, 12\rangle \\
 &+ w_2 (2 w_1 + w_2)|{1}, {2}, 11\rangle 
 +w_2 (w_1 + w_2) |{1}, 11, {2}\rangle,\\
|\xi_3(2,2)\rangle&=
w_2 (w_1 + w_2)|1, 1, 22\rangle 
+ w_2 (w_1 + w_2)|1, 2, 12\rangle 
+ w_2^2 |1, 12, 2\rangle + w_2^2 |2, 2, 11\rangle\\ 
&+ (w_1^2 + w_1 w_2 + w_2^2)|\emptyset, 1, 122\rangle 
 + w_2 (2 w_1 + w_2) |\emptyset, 2, 112\rangle 
+ w_2^2 |\emptyset, 11, 22\rangle \\
&+ w_2 (w_1 + w_2)|\emptyset, 12, 12\rangle 
+ 
 w_2 (2 w_1 + w_2)|\emptyset, 22, 11\rangle\\ 
 &+ w_2^2|\emptyset, 112, 2\rangle + 
 (2 w_1^2 + 2 w_1 w_2 + w_2^2)|\emptyset, 122, 1\rangle
+ (3 w_1^2 + 2 w_1 w_2 + w_2^2)|\emptyset, \emptyset, 1122\rangle,\\
|\xi_2(1,3)\rangle&=
 w_2|1, 222\rangle 
 +  w_2 |2, 122\rangle +  w_2 |12, 22\rangle 
 + (w_1 + w_2) |\emptyset, 1222\rangle,\\
|\xi_3(1,3)\rangle&=
w_2^2 |1, 2, 22\rangle 
+ w_2^2 |1, 22, 2\rangle 
+ w_2^2 |2, 2, 12\rangle + 
 w_2^2 |\emptyset, 1, 222\rangle 
 + w_2 (w_1 + w_2) |\emptyset, 2, 122\rangle \\
 &+ w_2^2 |\emptyset, 12, 22\rangle 
 + w_2 (w_1 + w_2) |\emptyset, 22, 12\rangle + 
 w_2^2 |\emptyset, 122, 2\rangle 
 +w_2 (w_1 + w_2) |\emptyset, 222, 1\rangle\\ 
 &+ (w_1^2 + w_1 w_2 + w_2^2) |\emptyset, \emptyset, 1222\rangle.
\end{align*}
For instance the both $|\mathscr{P}_2(1,2)\rangle
= |\xi_2(1,2)\rangle +C|\xi_2(1,2)\rangle$ and 
$|\mathscr{P}_2(2,1)\rangle
= |\xi_2(2,1)\rangle +C|\xi_2(2,1)\rangle$ 
consist of the six states with the following 
probability and the transition rate.

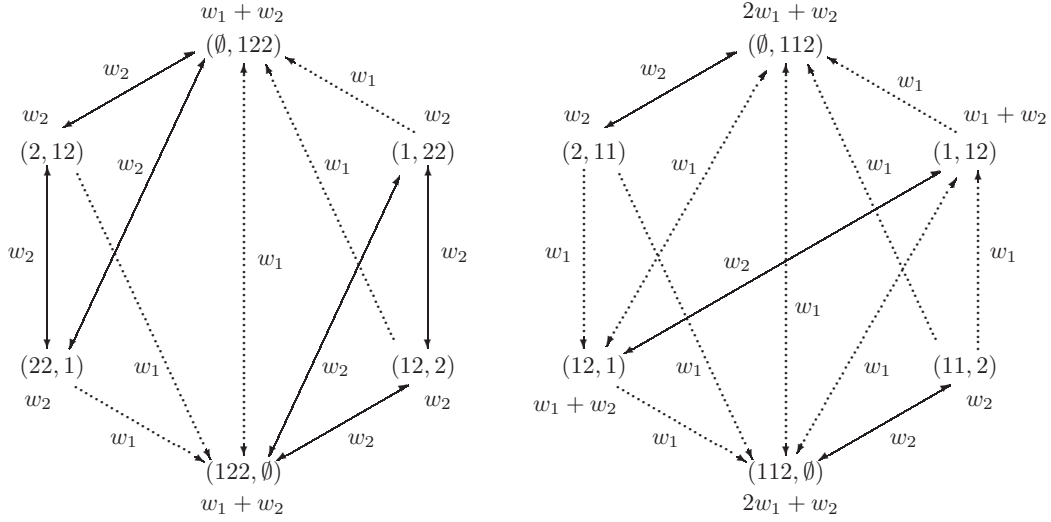
\begin{equation*}
\begin{picture}(0,190)(-20,19)

\put(-120,20){
\scalebox{0.9}{
\put(0,15){
\put(-20,180){$(\emptyset,122)$}
\put(58,135){$(1,22)$}
\put(58,45){$(12,2)$}
\put(-98,45){$(22,1)$}
\put(-98,135){$(2,12)$}
\put(-20,0){$(122,\emptyset)$}

\put(5,0){
\put(-8.7,171){\vector(0,1){4}}\put(-10,167.5){$\cdot$}\put(-10,165){$\cdot$}\put(-10,162.5){$\cdot$}\put(-10,160){$\cdot$}\put(-10,157.5){$\cdot$}\put(-10,155){$\cdot$}\put(-10,152.5){$\cdot$}\put(-10,150){$\cdot$}\put(-10,147.5){$\cdot$}\put(-10,145){$\cdot$}\put(-10,142.5){$\cdot$}\put(-10,140){$\cdot$}\put(-10,137.5){$\cdot$}
\put(-10,135){$\cdot$}\put(-10,132.5){$\cdot$}\put(-10,130){$\cdot$}\put(-10,127.5){$\cdot$}\put(-10,125){$\cdot$}\put(-10,122.5){$\cdot$}\put(-10,120){$\cdot$}\put(-10,117.5){$\cdot$}\put(-10,115){$\cdot$}\put(-10,112.5){$\cdot$}\put(-10,110){$\cdot$}\put(-10,107.5){$\cdot$}\put(-10,105){$\cdot$}\put(-10,102.5){$\cdot$}\put(-10,100){$\cdot$}\put(-10,97.5){$\cdot$}\put(-10,95){$\cdot$}\put(-10,92.5){$\cdot$}\put(-10,90){$\cdot$}\put(-10,87.5){$\cdot$}\put(-10,85){$\cdot$}\put(-10,82.5){$\cdot$}\put(-10,80){$\cdot$}\put(-10,77.5){$\cdot$}\put(-10,75){$\cdot$}\put(-10,72.5){$\cdot$}\put(-10,70){$\cdot$}\put(-10,67.5){$\cdot$}\put(-10,65){$\cdot$}\put(-10,62.5){$\cdot$}\put(-10,60){$\cdot$}\put(-10,57.5){$\cdot$}\put(-10,55){$\cdot$}\put(-10,52.5){$\cdot$}\put(-10,50){$\cdot$}\put(-10,47.5){$\cdot$}\put(-10,45){$\cdot$}\put(-10,42.5){$\cdot$}\put(-10,40){$\cdot$}\put(-10,37.5){$\cdot$}\put(-10,35){$\cdot$}\put(-10,32.5){$\cdot$}\put(-10,30){$\cdot$}\put(-10,27.5){$\cdot$}\put(-10,25){$\cdot$}\put(-10,22.5){$\cdot$}\put(-10,20){$\cdot$}\put(-10,17.5){$\cdot$}\put(-8.7,15){\vector(0,-1){5}}\put(-10,15){$\cdot$}}

{
\put(-80,150){\rotatebox{-60}{\vector(0,-1){4}}}
\put(-85,150){
\rotatebox{30}{\put(0,0){\vector(1,0){60}}}}
}

{
\put(-86.5,130){\vector(0,-1){75}}
\put(-179.6,258){
\rotatebox{180}{\put(-90,130){\vector(0,-1){4}}}}
}

\put(-8,7){
\rotatebox{5}{
\put(-90,55){
\rotatebox{60}{\put(0,0){\vector(1,0){135}}}}
\put(-157.3,113.5){
\rotatebox{180}{
\put(-90,55){
\rotatebox{60}{\put(0,0){\vector(1,0){4}}}}}}
}}

\put(76,-38)
{
\rotatebox{5}{
\put(-90,55){
\rotatebox{60}{\put(0,0){\vector(1,0){130}}}}
\put(-157.3,113.5){
\rotatebox{180}{
\put(-90,55){
\rotatebox{60}{\put(0,0){\vector(1,0){4}}}}}}
}
}

\put(90,-140)
{
\put(-80,150){\rotatebox{-60}{\vector(0,-1){4}}}
\put(-85,150){
\rotatebox{30}{\put(0,0){\vector(1,0){60}}}}
}

\put(160,0)
{
\put(-86.5,130){\vector(0,-1){75}}
\put(-179.6,258){
\rotatebox{180}{\put(-90,130){\vector(0,-1){4}}}}
}

\put(-75,10){
\rotatebox{60}{
\put(-10,70){$\cdot$}\put(-10,67.5){$\cdot$}\put(-10,65){$\cdot$}\put(-10,62.5){$\cdot$}\put(-10,60){$\cdot$}\put(-10,57.5){$\cdot$}\put(-10,55){$\cdot$}\put(-10,52.5){$\cdot$}\put(-10,50){$\cdot$}\put(-10,47.5){$\cdot$}\put(-10,45){$\cdot$}\put(-10,42.5){$\cdot$}\put(-10,40){$\cdot$}\put(-10,37.5){$\cdot$}\put(-10,35){$\cdot$}\put(-10,32.5){$\cdot$}\put(-10,30){$\cdot$}\put(-10,27.5){$\cdot$}\put(-10,25){$\cdot$}\put(-10,22.5){$\cdot$}\put(-10,20){$\cdot$}\put(-10,17.5){$\cdot$}\put(-8.7,15){\vector(0,-1){5}}\put(-10,15){$\cdot$}}
}

\put(95,140){

\put(-105,45){
\rotatebox{180}{
\put(-75,10){
\rotatebox{60}{
\put(-10,70){$\cdot$}\put(-10,67.5){$\cdot$}\put(-10,65){$\cdot$}\put(-10,62.5){$\cdot$}\put(-10,60){$\cdot$}\put(-10,57.5){$\cdot$}\put(-10,55){$\cdot$}\put(-10,52.5){$\cdot$}\put(-10,50){$\cdot$}\put(-10,47.5){$\cdot$}\put(-10,45){$\cdot$}\put(-10,42.5){$\cdot$}\put(-10,40){$\cdot$}\put(-10,37.5){$\cdot$}\put(-10,35){$\cdot$}\put(-10,32.5){$\cdot$}\put(-10,30){$\cdot$}\put(-10,27.5){$\cdot$}\put(-10,25){$\cdot$}\put(-10,22.5){$\cdot$}\put(-10,20){$\cdot$}\put(-10,17.5){$\cdot$}\put(-8.7,15){\vector(0,-1){5}}\put(-10,15){$\cdot$}}
}}}
}

\put(-70,3){
\rotatebox{25}{
\put(-10,140){$\cdot$}\put(-10,137.5){$\cdot$}
\put(-10,135){$\cdot$}\put(-10,132.5){$\cdot$}\put(-10,130){$\cdot$}\put(-10,127.5){$\cdot$}\put(-10,125){$\cdot$}\put(-10,122.5){$\cdot$}\put(-10,120){$\cdot$}\put(-10,117.5){$\cdot$}\put(-10,115){$\cdot$}\put(-10,112.5){$\cdot$}\put(-10,110){$\cdot$}\put(-10,107.5){$\cdot$}\put(-10,105){$\cdot$}\put(-10,102.5){$\cdot$}\put(-10,100){$\cdot$}\put(-10,97.5){$\cdot$}\put(-10,95){$\cdot$}\put(-10,92.5){$\cdot$}\put(-10,90){$\cdot$}\put(-10,87.5){$\cdot$}\put(-10,85){$\cdot$}\put(-10,82.5){$\cdot$}\put(-10,80){$\cdot$}\put(-10,77.5){$\cdot$}\put(-10,75){$\cdot$}\put(-10,72.5){$\cdot$}\put(-10,70){$\cdot$}\put(-10,67.5){$\cdot$}\put(-10,65){$\cdot$}\put(-10,62.5){$\cdot$}\put(-10,60){$\cdot$}\put(-10,57.5){$\cdot$}\put(-10,55){$\cdot$}\put(-10,52.5){$\cdot$}\put(-10,50){$\cdot$}\put(-10,47.5){$\cdot$}\put(-10,45){$\cdot$}\put(-10,42.5){$\cdot$}\put(-10,40){$\cdot$}\put(-10,37.5){$\cdot$}\put(-10,35){$\cdot$}\put(-10,32.5){$\cdot$}\put(-10,30){$\cdot$}\put(-10,27.5){$\cdot$}\put(-10,25){$\cdot$}\put(-10,22.5){$\cdot$}\put(-10,20){$\cdot$}\put(-10,17.5){$\cdot$}\put(-8.7,15){\vector(0,-1){5}}\put(-10,15){$\cdot$}}}

\put(-10,180){
\rotatebox{180}{
\rotatebox{25}{
\put(-10,135){$\cdot$}\put(-10,132.5){$\cdot$}\put(-10,130){$\cdot$}\put(-10,127.5){$\cdot$}\put(-10,125){$\cdot$}\put(-10,122.5){$\cdot$}\put(-10,120){$\cdot$}\put(-10,117.5){$\cdot$}\put(-10,115){$\cdot$}\put(-10,112.5){$\cdot$}\put(-10,110){$\cdot$}\put(-10,107.5){$\cdot$}\put(-10,105){$\cdot$}\put(-10,102.5){$\cdot$}\put(-10,100){$\cdot$}\put(-10,97.5){$\cdot$}\put(-10,95){$\cdot$}\put(-10,92.5){$\cdot$}\put(-10,90){$\cdot$}\put(-10,87.5){$\cdot$}\put(-10,85){$\cdot$}\put(-10,82.5){$\cdot$}\put(-10,80){$\cdot$}\put(-10,77.5){$\cdot$}\put(-10,75){$\cdot$}\put(-10,72.5){$\cdot$}\put(-10,70){$\cdot$}\put(-10,67.5){$\cdot$}\put(-10,65){$\cdot$}\put(-10,62.5){$\cdot$}\put(-10,60){$\cdot$}\put(-10,57.5){$\cdot$}\put(-10,55){$\cdot$}\put(-10,52.5){$\cdot$}\put(-10,50){$\cdot$}\put(-10,47.5){$\cdot$}\put(-10,45){$\cdot$}\put(-10,42.5){$\cdot$}\put(-10,40){$\cdot$}\put(-10,37.5){$\cdot$}\put(-10,35){$\cdot$}\put(-10,32.5){$\cdot$}\put(-10,30){$\cdot$}\put(-10,27.5){$\cdot$}\put(-10,25){$\cdot$}\put(-10,22.5){$\cdot$}\put(-10,20){$\cdot$}\put(-10,17.5){$\cdot$}\put(-8.7,15){\vector(0,-1){5}}\put(-10,15){$\cdot$}}}}
}

\put(-22,2){$w_1+w_2$}
\put(-22,209){$w_1+w_2$}
\put(72,46){$w_2$}
\put(72,166){$w_2$}
\put(-96,46){$w_2$}
\put(-97,166){$w_2$}

\put(-63,185){$w_2$}
\put(-103,108){$w_2$}
\put(-57,145){$w_2$}
\put(79,108){$w_2$}
\put(40,30){$w_2$}
\put(29,60){$w_2$}

\put(41,182){$w_1$}
\put(30,145){$w_1$}
\put(2,105){$w_1$}
\put(-50,60){$w_1$}
\put(-60,30){$w_1$}

}

}

\put(85,20){
\scalebox{0.9}{
\put(0,15){
\put(-20,180){$(\emptyset,112)$}
\put(58,135){$(1,12)$}
\put(58,45){$(11,2)$}
\put(-98,45){$(12,1)$}
\put(-98,135){$(2,11)$}
\put(-20,0){$(112,\emptyset)$}

\put(5,0){
\put(-8.7,171){\vector(0,1){4}}\put(-10,167.5){$\cdot$}\put(-10,165){$\cdot$}\put(-10,162.5){$\cdot$}\put(-10,160){$\cdot$}\put(-10,157.5){$\cdot$}\put(-10,155){$\cdot$}\put(-10,152.5){$\cdot$}\put(-10,150){$\cdot$}\put(-10,147.5){$\cdot$}\put(-10,145){$\cdot$}\put(-10,142.5){$\cdot$}\put(-10,140){$\cdot$}\put(-10,137.5){$\cdot$}
\put(-10,135){$\cdot$}\put(-10,132.5){$\cdot$}\put(-10,130){$\cdot$}\put(-10,127.5){$\cdot$}\put(-10,125){$\cdot$}\put(-10,122.5){$\cdot$}\put(-10,120){$\cdot$}\put(-10,117.5){$\cdot$}\put(-10,115){$\cdot$}\put(-10,112.5){$\cdot$}\put(-10,110){$\cdot$}\put(-10,107.5){$\cdot$}\put(-10,105){$\cdot$}\put(-10,102.5){$\cdot$}\put(-10,100){$\cdot$}\put(-10,97.5){$\cdot$}\put(-10,95){$\cdot$}\put(-10,92.5){$\cdot$}\put(-10,90){$\cdot$}\put(-10,87.5){$\cdot$}\put(-10,85){$\cdot$}\put(-10,82.5){$\cdot$}\put(-10,80){$\cdot$}\put(-10,77.5){$\cdot$}\put(-10,75){$\cdot$}\put(-10,72.5){$\cdot$}\put(-10,70){$\cdot$}\put(-10,67.5){$\cdot$}\put(-10,65){$\cdot$}\put(-10,62.5){$\cdot$}\put(-10,60){$\cdot$}\put(-10,57.5){$\cdot$}\put(-10,55){$\cdot$}\put(-10,52.5){$\cdot$}\put(-10,50){$\cdot$}\put(-10,47.5){$\cdot$}\put(-10,45){$\cdot$}\put(-10,42.5){$\cdot$}\put(-10,40){$\cdot$}\put(-10,37.5){$\cdot$}\put(-10,35){$\cdot$}\put(-10,32.5){$\cdot$}\put(-10,30){$\cdot$}\put(-10,27.5){$\cdot$}\put(-10,25){$\cdot$}\put(-10,22.5){$\cdot$}\put(-10,20){$\cdot$}\put(-10,17.5){$\cdot$}\put(-8.7,15){\vector(0,-1){5}}\put(-10,15){$\cdot$}}

{
\put(-80,150){\rotatebox{-60}{\vector(0,-1){4}}}
\put(-85,150){
\rotatebox{30}{\put(0,0){\vector(1,0){60}}}}
}

\put(-80,45){
\put(-10,82.5){$\cdot$}\put(-10,80){$\cdot$}\put(-10,77.5){$\cdot$}\put(-10,75){$\cdot$}\put(-10,72.5){$\cdot$}\put(-10,70){$\cdot$}\put(-10,67.5){$\cdot$}\put(-10,65){$\cdot$}\put(-10,62.5){$\cdot$}\put(-10,60){$\cdot$}\put(-10,57.5){$\cdot$}\put(-10,55){$\cdot$}\put(-10,52.5){$\cdot$}\put(-10,50){$\cdot$}\put(-10,47.5){$\cdot$}\put(-10,45){$\cdot$}\put(-10,42.5){$\cdot$}\put(-10,40){$\cdot$}\put(-10,37.5){$\cdot$}\put(-10,35){$\cdot$}\put(-10,32.5){$\cdot$}\put(-10,30){$\cdot$}\put(-10,27.5){$\cdot$}\put(-10,25){$\cdot$}\put(-10,22.5){$\cdot$}\put(-10,20){$\cdot$}\put(-10,17.5){$\cdot$}\put(-8.7,15){\vector(0,-1){5}}\put(-10,15){$\cdot$}}

\put(-90,185){
\rotatebox{150}{
\put(-8.5,140){\vector(0,1){4}}\put(-10,137.5){$\cdot$}
\put(-10,135){$\cdot$}\put(-10,132.5){$\cdot$}\put(-10,130){$\cdot$}\put(-10,127.5){$\cdot$}\put(-10,125){$\cdot$}\put(-10,122.5){$\cdot$}\put(-10,120){$\cdot$}\put(-10,117.5){$\cdot$}\put(-10,115){$\cdot$}\put(-10,112.5){$\cdot$}\put(-10,110){$\cdot$}\put(-10,107.5){$\cdot$}\put(-10,105){$\cdot$}\put(-10,102.5){$\cdot$}\put(-10,100){$\cdot$}\put(-10,97.5){$\cdot$}\put(-10,95){$\cdot$}\put(-10,92.5){$\cdot$}\put(-10,90){$\cdot$}\put(-10,87.5){$\cdot$}\put(-10,85){$\cdot$}\put(-10,82.5){$\cdot$}\put(-10,80){$\cdot$}\put(-10,77.5){$\cdot$}\put(-10,75){$\cdot$}\put(-10,72.5){$\cdot$}\put(-10,70){$\cdot$}\put(-10,67.5){$\cdot$}\put(-10,65){$\cdot$}\put(-10,62.5){$\cdot$}\put(-10,60){$\cdot$}\put(-10,57.5){$\cdot$}\put(-10,55){$\cdot$}\put(-10,52.5){$\cdot$}\put(-10,50){$\cdot$}\put(-10,47.5){$\cdot$}\put(-10,45){$\cdot$}\put(-10,42.5){$\cdot$}\put(-10,40){$\cdot$}\put(-10,37.5){$\cdot$}\put(-10,35){$\cdot$}\put(-10,32.5){$\cdot$}\put(-10,30){$\cdot$}\put(-10,27.5){$\cdot$}\put(-10,25){$\cdot$}\put(-10,22.5){$\cdot$}\put(-10,20){$\cdot$}\put(-10,17.5){$\cdot$}\put(-8.7,15){\vector(0,-1){5}}\put(-10,15){$\cdot$}}}

\put(-10,140){
\rotatebox{150}{
\put(-8.5,140){\vector(0,1){4}}\put(-10,137.5){$\cdot$}
\put(-10,135){$\cdot$}\put(-10,132.5){$\cdot$}\put(-10,130){$\cdot$}\put(-10,127.5){$\cdot$}\put(-10,125){$\cdot$}\put(-10,122.5){$\cdot$}\put(-10,120){$\cdot$}\put(-10,117.5){$\cdot$}\put(-10,115){$\cdot$}\put(-10,112.5){$\cdot$}\put(-10,110){$\cdot$}\put(-10,107.5){$\cdot$}\put(-10,105){$\cdot$}\put(-10,102.5){$\cdot$}\put(-10,100){$\cdot$}\put(-10,97.5){$\cdot$}\put(-10,95){$\cdot$}\put(-10,92.5){$\cdot$}\put(-10,90){$\cdot$}\put(-10,87.5){$\cdot$}\put(-10,85){$\cdot$}\put(-10,82.5){$\cdot$}\put(-10,80){$\cdot$}\put(-10,77.5){$\cdot$}\put(-10,75){$\cdot$}\put(-10,72.5){$\cdot$}\put(-10,70){$\cdot$}\put(-10,67.5){$\cdot$}\put(-10,65){$\cdot$}\put(-10,62.5){$\cdot$}\put(-10,60){$\cdot$}\put(-10,57.5){$\cdot$}\put(-10,55){$\cdot$}\put(-10,52.5){$\cdot$}\put(-10,50){$\cdot$}\put(-10,47.5){$\cdot$}\put(-10,45){$\cdot$}\put(-10,42.5){$\cdot$}\put(-10,40){$\cdot$}\put(-10,37.5){$\cdot$}\put(-10,35){$\cdot$}\put(-10,32.5){$\cdot$}\put(-10,30){$\cdot$}\put(-10,27.5){$\cdot$}\put(-10,25){$\cdot$}\put(-10,22.5){$\cdot$}\put(-10,20){$\cdot$}\put(-10,17.5){$\cdot$}\put(-8.7,15){\vector(0,-1){5}}\put(-10,15){$\cdot$}}}

\put(8,-95)
{
\put(-80,150){\rotatebox{-60}{\vector(0,-1){4}}}
\put(-85,150){
\rotatebox{30}{\put(0,0){\vector(1,0){150}}}}
}

\put(90,-140)
{
\put(-80,150){\rotatebox{-60}{\vector(0,-1){4}}}
\put(-85,150){
\rotatebox{30}{\put(0,0){\vector(1,0){60}}}}
}

\put(-10,185){
\rotatebox{180}{
\put(-75,45){
\put(-10,82.5){$\cdot$}\put(-10,80){$\cdot$}\put(-10,77.5){$\cdot$}\put(-10,75){$\cdot$}\put(-10,72.5){$\cdot$}\put(-10,70){$\cdot$}\put(-10,67.5){$\cdot$}\put(-10,65){$\cdot$}\put(-10,62.5){$\cdot$}\put(-10,60){$\cdot$}\put(-10,57.5){$\cdot$}\put(-10,55){$\cdot$}\put(-10,52.5){$\cdot$}\put(-10,50){$\cdot$}\put(-10,47.5){$\cdot$}\put(-10,45){$\cdot$}\put(-10,42.5){$\cdot$}\put(-10,40){$\cdot$}\put(-10,37.5){$\cdot$}\put(-10,35){$\cdot$}\put(-10,32.5){$\cdot$}\put(-10,30){$\cdot$}\put(-10,27.5){$\cdot$}\put(-10,25){$\cdot$}\put(-10,22.5){$\cdot$}\put(-10,20){$\cdot$}\put(-10,17.5){$\cdot$}\put(-8.7,15){\vector(0,-1){5}}\put(-10,15){$\cdot$}}}}

\put(-75,10){
\rotatebox{60}{
\put(-10,70){$\cdot$}\put(-10,67.5){$\cdot$}\put(-10,65){$\cdot$}\put(-10,62.5){$\cdot$}\put(-10,60){$\cdot$}\put(-10,57.5){$\cdot$}\put(-10,55){$\cdot$}\put(-10,52.5){$\cdot$}\put(-10,50){$\cdot$}\put(-10,47.5){$\cdot$}\put(-10,45){$\cdot$}\put(-10,42.5){$\cdot$}\put(-10,40){$\cdot$}\put(-10,37.5){$\cdot$}\put(-10,35){$\cdot$}\put(-10,32.5){$\cdot$}\put(-10,30){$\cdot$}\put(-10,27.5){$\cdot$}\put(-10,25){$\cdot$}\put(-10,22.5){$\cdot$}\put(-10,20){$\cdot$}\put(-10,17.5){$\cdot$}\put(-8.7,15){\vector(0,-1){5}}\put(-10,15){$\cdot$}}
}

\put(95,140){

\put(-105,45){
\rotatebox{180}{
\put(-75,10){
\rotatebox{60}{
\put(-10,70){$\cdot$}\put(-10,67.5){$\cdot$}\put(-10,65){$\cdot$}\put(-10,62.5){$\cdot$}\put(-10,60){$\cdot$}\put(-10,57.5){$\cdot$}\put(-10,55){$\cdot$}\put(-10,52.5){$\cdot$}\put(-10,50){$\cdot$}\put(-10,47.5){$\cdot$}\put(-10,45){$\cdot$}\put(-10,42.5){$\cdot$}\put(-10,40){$\cdot$}\put(-10,37.5){$\cdot$}\put(-10,35){$\cdot$}\put(-10,32.5){$\cdot$}\put(-10,30){$\cdot$}\put(-10,27.5){$\cdot$}\put(-10,25){$\cdot$}\put(-10,22.5){$\cdot$}\put(-10,20){$\cdot$}\put(-10,17.5){$\cdot$}\put(-8.7,15){\vector(0,-1){5}}\put(-10,15){$\cdot$}}
}}}
}

\put(-70,3){
\rotatebox{25}{
\put(-10,140){$\cdot$}\put(-10,137.5){$\cdot$}
\put(-10,135){$\cdot$}\put(-10,132.5){$\cdot$}\put(-10,130){$\cdot$}\put(-10,127.5){$\cdot$}\put(-10,125){$\cdot$}\put(-10,122.5){$\cdot$}\put(-10,120){$\cdot$}\put(-10,117.5){$\cdot$}\put(-10,115){$\cdot$}\put(-10,112.5){$\cdot$}\put(-10,110){$\cdot$}\put(-10,107.5){$\cdot$}\put(-10,105){$\cdot$}\put(-10,102.5){$\cdot$}\put(-10,100){$\cdot$}\put(-10,97.5){$\cdot$}\put(-10,95){$\cdot$}\put(-10,92.5){$\cdot$}\put(-10,90){$\cdot$}\put(-10,87.5){$\cdot$}\put(-10,85){$\cdot$}\put(-10,82.5){$\cdot$}\put(-10,80){$\cdot$}\put(-10,77.5){$\cdot$}\put(-10,75){$\cdot$}\put(-10,72.5){$\cdot$}\put(-10,70){$\cdot$}\put(-10,67.5){$\cdot$}\put(-10,65){$\cdot$}\put(-10,62.5){$\cdot$}\put(-10,60){$\cdot$}\put(-10,57.5){$\cdot$}\put(-10,55){$\cdot$}\put(-10,52.5){$\cdot$}\put(-10,50){$\cdot$}\put(-10,47.5){$\cdot$}\put(-10,45){$\cdot$}\put(-10,42.5){$\cdot$}\put(-10,40){$\cdot$}\put(-10,37.5){$\cdot$}\put(-10,35){$\cdot$}\put(-10,32.5){$\cdot$}\put(-10,30){$\cdot$}\put(-10,27.5){$\cdot$}\put(-10,25){$\cdot$}\put(-10,22.5){$\cdot$}\put(-10,20){$\cdot$}\put(-10,17.5){$\cdot$}\put(-8.7,15){\vector(0,-1){5}}\put(-10,15){$\cdot$}}}

\put(-10,180){
\rotatebox{180}{
\rotatebox{25}{
\put(-10,135){$\cdot$}\put(-10,132.5){$\cdot$}\put(-10,130){$\cdot$}\put(-10,127.5){$\cdot$}\put(-10,125){$\cdot$}\put(-10,122.5){$\cdot$}\put(-10,120){$\cdot$}\put(-10,117.5){$\cdot$}\put(-10,115){$\cdot$}\put(-10,112.5){$\cdot$}\put(-10,110){$\cdot$}\put(-10,107.5){$\cdot$}\put(-10,105){$\cdot$}\put(-10,102.5){$\cdot$}\put(-10,100){$\cdot$}\put(-10,97.5){$\cdot$}\put(-10,95){$\cdot$}\put(-10,92.5){$\cdot$}\put(-10,90){$\cdot$}\put(-10,87.5){$\cdot$}\put(-10,85){$\cdot$}\put(-10,82.5){$\cdot$}\put(-10,80){$\cdot$}\put(-10,77.5){$\cdot$}\put(-10,75){$\cdot$}\put(-10,72.5){$\cdot$}\put(-10,70){$\cdot$}\put(-10,67.5){$\cdot$}\put(-10,65){$\cdot$}\put(-10,62.5){$\cdot$}\put(-10,60){$\cdot$}\put(-10,57.5){$\cdot$}\put(-10,55){$\cdot$}\put(-10,52.5){$\cdot$}\put(-10,50){$\cdot$}\put(-10,47.5){$\cdot$}\put(-10,45){$\cdot$}\put(-10,42.5){$\cdot$}\put(-10,40){$\cdot$}\put(-10,37.5){$\cdot$}\put(-10,35){$\cdot$}\put(-10,32.5){$\cdot$}\put(-10,30){$\cdot$}\put(-10,27.5){$\cdot$}\put(-10,25){$\cdot$}\put(-10,22.5){$\cdot$}\put(-10,20){$\cdot$}\put(-10,17.5){$\cdot$}\put(-8.7,15){\vector(0,-1){5}}\put(-10,15){$\cdot$}}}}
}

\put(-22,2){$2w_1+w_2$}
\put(-22,209){$2w_1+w_2$}
\put(72,46){$w_2$}
\put(71,166){$w_1+w_2$}
\put(-110,44){$w_1+w_2$}
\put(-97,166){$w_2$}

\put(-65,185){$w_2$}
\put(-105,108){$w_1$}
\put(-52,145){$w_1$}
\put(82,108){$w_1$}
\put(40,30){$w_2$}
\put(29,60){$w_1$}

\put(43,180){$w_1$}
\put(30,145){$w_1$}
\put(0,85){$w_1$}
\put(-50,60){$w_1$}
\put(-60,30){$w_1$}
\put(-30,104){$w_2$}
}
}
\end{picture}
\end{equation*}
Dashed and solid arrows denote the transitions with rate $w_1$  and $w_2$,
respectively.
One can check the steady state condition directly from these diagrams.
 
\vspace{0.2cm}
For the 3-iTAZRP one has
\begin{align*}
|\xi_2(1,1,1)\rangle&=
w_2 (w_1 + w_2 + w_3) |1, 23\rangle 
+w_2 w_3 |2, 13\rangle 
+ (w_1 + w_2) w_3 |3, 12\rangle \\
&+ (w_1 + w_2) (w_1 + w_2 + w_3) |\emptyset, 123\rangle,\\
|\xi_3(1,1,1)\rangle&=
w_2^2 w_3^2 |1, 2, 3\rangle 
+ w_2 w_3 (w_1^2 + w_1 w_2 + w_2^2 + w_1 w_3 + w_2 w_3)|1, 3, 2\rangle \\
&+w_2^2 (w_1^2 + w_1 w_2 + w_2^2 + w_1 w_3 + w_2 w_3 + w_3^2) |\emptyset, 1, 23\rangle \\
&+ w_2 w_3 (w_1 w_2 + w_1 w_3 + w_2 w_3) |\emptyset, 2, 13\rangle 
+ (w_1^2 + w_1 w_2 + w_2^2) w_3 (w_1 + w_2 + w_3) |\emptyset, 3, 12\rangle \\
&+ (w_1^2 + w_1 w_2 + w_2^2) w_3^2 |\emptyset, 12, 3\rangle
+ w_2^2 w_3 (w_1 + w_2 + w_3) |\emptyset, 13, 2\rangle\\ 
&+ w_2 (w_1 + w_2) (w_1^2 + w_1 w_2 + w_2^2 + w_1 w_3 + w_2 w_3 + w_3^2) |\emptyset, 23, 1\rangle \\
&+ (w_1^2 + w_1 w_2 + w_2^2) (w_1^2 + w_1 w_2 + w_2^2 + 
   w_1 w_3 + w_2 w_3 + w_3^2) |\emptyset, \emptyset, 123\rangle,\\
|\xi_2(2,1,1)\rangle&=
(w_1 + w_2) (2 w_1 + w_2 + w_3) |1, 123\rangle 
+ w_2 w_3|2, 113\rangle 
+ (2 w_1 + w_2) w_3 |3, 112\rangle \\
&+ w_2 (2 w_1 + w_2 + w_3) |11, 23\rangle + 
 (w_1 + w_2) w_3|12, 13\rangle \\
& + (2 w_1 + w_2) (2 w_1 + w_2 + w_3)|\emptyset, 1123\rangle,\\
|\xi_2(1,2,1)\rangle&=
w_2 (w_1 + 2 w_2 + w_3) |1, 223\rangle 
+ w_2 (w_2 + w_3) |2, 123\rangle 
+ (w_1 + w_2) w_3 |3, 122\rangle \\
&+ w_2 (w_1 + w_2 + w_3) |12, 23\rangle + 
 w_2 w_3|13, 22\rangle 
 + (w_1 + w_2) (w_1 + 2 w_2 + w_3) |\emptyset, 1223\rangle,\\
|\xi_2(1,1,2)\rangle&=
w_2 (w_1 + w_2 + w_3) |1, 233\rangle 
+ w_2 w_3|2, 133\rangle 
+ (w_1 + w_2) w_3 |3, 123\rangle \\
&+ (w_1 + w_2) w_3 |12, 33\rangle + 
w_2 w_3 |13, 23\rangle + (w_1 + w_2) (w_1 + w_2 + w_3) |\emptyset, 1233\rangle.
\end{align*}

For the homogeneous case where all the $w_i$'s are equal,
the states having the largest probability are 
$C^i|\emptyset, \ldots \emptyset, \text{all}\rangle\,(i \in \Z_L)$,
which is a symptom of {\em condensation} \cite{EH,GSS}.
See \cite[eq.(4.11)]{KMO3}.
\end{example}

\section{Matrix product formula}\label{sec:mpf}
Let $F = \bigoplus_{m\ge 0} \C |m\rangle$ and 
$F^\ast = \bigoplus_{m\ge 0} \C \langle m |$ be the Fock space and its dual
with the bilinear pairing $\langle m | m' \rangle = \delta^{m'}_{m}$.
Let 
${\bf 1}, {\bf a}^{\pm}, {\bf k}$ and 
${\bf d}={\bf 1}-{\bf k}$ be the linear operators acting on them by
\begin{align*}
&{\bf 1}|m\rangle = |m\rangle,\quad 
{\bf a}^+|m\rangle = |m+1\rangle, \quad
{\bf a}^-|m\rangle = |m-1\rangle,\quad
{\bf k}|m\rangle = \delta_m^0|m\rangle,\\
&\langle m |{\bf 1}= \langle m |,\quad 
\langle m |{\bf a}^+ = \langle m-1 |, \quad
\langle m |{\bf a}^- = \langle m +1|,\quad
\langle m |{\bf k} = \delta_m^0\langle m |
\end{align*}
with $|\!-\!1\rangle = 0$ and 
$\langle\!-1| = 0$.
They satisfy 
$(\langle m | X)|m'\rangle = \langle m| (X|m'\rangle)$.
We denote the trace over $F$ by 
$\mathrm{Tr}(X)=\sum_{m \ge 0}\langle m| X | m\rangle$.
The trace over $F^{\otimes N}$ is the product of the one on 
each component.
In our working below, its convergence will always be assured by
the fact that the relevant $X$ contains ${\bf k}^{\otimes N}$ as an overall factor. 

Let $\mu=(\mu^1,\ldots, \mu^{n-1}) \in (\Z_{\ge 0})^{n-1}$ and 
$\alpha=(\alpha^1,\ldots, \alpha^n) \in (\Z_{\ge 0})^{n}$ be local states of 
$(n-1)$-iTAZRP and $n$-iTAZRP in multiplicity representation, respectively.
We define the operators 
$A_{\mu,\alpha}=A^{(n)}_{\mu,\alpha}, 
\hat{A}_{\mu,\alpha}=\hat{A}^{(n)}_{\mu,\alpha}
\in \mathrm{End}(F^{\otimes n-1})$ 
for $n \ge 2$ by
\begin{align}
A_{\mu,\alpha} &= P_+(\mu) 
\Bigl(\;\sum_{r=1}^{n-1}
\delta_{\alpha^{r+1}, \ldots,  \alpha^n}^{\; \,0,\, . . . \ldots, \, 0}
w_r K_r+ w_nK_n\, \Bigr)P_-(\overline{\alpha}),
\label{adef}\\
\hat{A}_{\mu,\alpha} &= P_+(\mu) 
\Bigl(\;\sum_{r=1}^{n-1}
\delta_{\alpha^{r+1}, \ldots,  \alpha^n}^{\; \,0,\, . . . \ldots, \, 0}
w_r(w_r+g(\alpha)) K_r+ w_n(w_n+g(\alpha))K_n\, \Bigr)P_-(\overline{\alpha}),
\label{ahdef}\\
P_\pm(\gamma) &= 
({\bf a}^\pm)^{\gamma^1} \otimes \cdots \otimes ({\bf a}^\pm)^{\gamma^{n-1}}
\;\;\text{for}\; \;\gamma=(\gamma^1,\ldots,\gamma^{n-1}) \in 
(\Z_{\ge 0})^{n-1},
\nonumber\\
K_r &=\overbrace{{\bf k} \otimes \cdots \otimes {\bf k}}^{r-1}\otimes 
{\bf d} \otimes 
\overbrace{{\bf 1} \otimes \cdots \otimes {\bf 1}}^{n-1-r}
\quad (1 \le r <n),\quad
K_n = {\bf k}^{\otimes n-1},
\nonumber
\end{align}
where $\overline{\alpha}=(\alpha^1,\ldots, \alpha^{n-1})$
and $g(\alpha)$ is defined in (\ref{gdef}).
We will exhibit the $n$-dependence as 
$ A^{(n)}_{\mu,\alpha}$ when preferable.
Note that $P_\pm(\gamma)$ and $K_r$ are also 
dependent on $n$ although it is not exhibited.
For $n=2, 3$  (\ref{adef}) looks as 
\begin{align*}
&A_{(\mu^1),(\alpha^1,\alpha^2)} = {\bf a}^{\mu^1}
(\delta^0_{\alpha^2}w_1{\bf d} + w_2{\bf k}){\bf a}^{-\alpha^1},\\
&A_{(\mu^1,\mu^2),(\alpha^1,\alpha^2,\alpha^3)} 
= ({\bf a}^{\mu^1}\otimes {\bf a}^{\mu^2})
(\delta^{0,\;\;0}_{\alpha^2,\alpha^3}w_1 {\bf d} \otimes {\bf 1} 
+ \delta^0_{\alpha^3}w_2 {\bf k}\otimes{\bf d}
+ w_3 {\bf k} \otimes {\bf k})
({\bf a}^{-\alpha^1}\otimes {\bf a}^{-\alpha^2})\\
&=\delta^{0,\;\;0}_{\alpha^2,\alpha^3}
w_1{\bf a}^{\mu^1}{\bf d}\,{\bf a}^{-\alpha^1}
\otimes {\bf a}^{\mu^2}
+\delta_{\alpha^3}^0w_2{\bf a}^{\mu^1}{\bf k}\,{\bf a}^{-\alpha^1}
\otimes{\bf a}^{\mu^2}{\bf d}\,{\bf a}^{-\alpha^2}
+w_3{\bf a}^{\mu^1}{\bf k}\,{\bf a}^{-\alpha^1}
\otimes{\bf a}^{\mu^2}{\bf k}\,{\bf a}^{-\alpha^2},
\end{align*}
where we have used the simplified notation
${\bf a}^{\pm m} = ({\bf a}^\pm)^m$ for $m \ge 0$.
In the homogeneous case, (\ref{adef}) simplifies to 
\begin{align*}
A_{\mu,\alpha} |_{w_1=\cdots = w_n=1} = P_+(\mu)(\,
{\bf k}^{\alpha^2+\cdots + \alpha^n} \!\otimes \cdots 
\otimes {\bf k}^{\alpha^{n-1}+\alpha^n} \!
\otimes {\bf k}^{\alpha^n})P_-(\overline{\alpha})
\end{align*}
due to the identity $\delta_i^0{\bf d} + {\bf k} = {\bf k}^i$. 

Given an $n$-iTAZRP local state 
$\sigma \in (\Z_{\ge 0})^n$ in multiplicity representation for $n \ge 2$, we define
\begin{align}\label{X}
X_\sigma = X^{(n)}_\sigma = \sum
A^{(2)}_{\mu^{(1)}, \mu^{(2)}} \otimes 
A^{(3)}_{\mu^{(2)}, \mu^{(3)}} \otimes \cdots \otimes
A^{(n)}_{\mu^{(n-1)}, \sigma} 
\in \mathrm{End}(F^{\otimes n(n-1)/2}),
\end{align}
where the sums range over 
$\mu^{(a)} \in (\Z_{\ge 0})^{a}$ for all $a \in [1,n-1]$. 
By the definition it satisfies the recursion relation
\begin{align}\label{xrec}
X^{(n)}_\sigma = \sum_{\mu \in (\Z_{\ge 0})^{n-1}}
X^{(n-1)}_\mu \otimes A^{(n)}_{\mu,\sigma},
\end{align}
where the $n=2$ case should be understood as
$X^{(2)}_\sigma=\sum_{\mu
\in \Z_{\ge 0}}A^{(2)}_{\mu,\sigma}$.
Now we state our first main result.
\begin{theorem}\label{th:P}
The steady state probability of the $n$-$\mathrm{iTAZRP}$ in basic sectors 
is given in the matrix product form
\begin{align*}
{\mathbb P}(\sigma_1, \ldots, \sigma_L) = (w_2\cdots w_n)^{-1}
\mathrm{Tr}\bigl(X_{\sigma_1} \cdots X_{\sigma_L}\bigr),
\end{align*}
where the trace is taken over $F^{\otimes n(n-1)/2}$.
\end{theorem}

Before the proof we check basic properties of the formula.
\begin{lemma}\label{le:iii} 
The quantity
$(w_2\cdots w_n)^{-1}
\mathrm{Tr}\bigl(X_{\sigma_1} \cdots X_{\sigma_L}\bigr)$ 
is finite and is a 
homogeneous polynomial in $w_1,\ldots, w_n$ of degree $(n-1)(L-1)$
with coefficients from $\Z_{\ge 0}$.
\end{lemma}
\begin{proof}
From (\ref{X}) the trace 
$\mathrm{Tr}\bigl(X_{\sigma_1} \cdots X_{\sigma_L}\bigr)$
is expanded as
\begin{align}\label{ryb}
\sum \mathrm{Tr}_F\Bigl(
A^{(2)}_{\mu^{(1)}_1, \mu^{(2)}_1}\cdots
A^{(2)}_{\mu^{(1)}_L, \mu^{(2)}_L}\Bigr)\cdots
\mathrm{Tr}_{F^{\otimes n-1}}\Bigl(
A^{(n)}_{\mu^{(n-1)}_1, \sigma_1}\cdots
A^{(n)}_{\mu^{(n-1)}_L, \sigma_L}\Bigr),
\end{align}
where the sums extend over
$\mu^{(a)}_i \in (\Z_{\ge 0})^a$
for each $(a, i) \in [1,n-1] \times \Z_L$.
Let us set $\mu^{(a)}_i =
(\mu^{(a),1}_i,\ldots, \mu^{(a),a}_i)$.
Since $P_+(\mu), P_-(\overline{\alpha})$ and $K_r$ in (\ref{adef}) are 
creation, annihilation and diagonal operators respectively,
the traces are non-vanishing only if 
$\sum_{i \in \Z_L}\mu^{(a-1),b}_i 
= \sum_{i \in \Z_L}\mu^{(a),b}_i$ for $1\le b < a \le n$.
Here we have set $\mu^{(n),b}_i =\sigma^b_i$
$(b \in [1,n])$, which is the multiplicity of $b$ in $\sigma_i$ as in (\ref{cie}).
For the configuration 
$(\sigma_1, \ldots, \sigma_L)$ in the basic sector $S(m_1,\ldots, m_n)$,
this leads to the constraint 
\begin{align}\label{cst}
\sum_{i \in \Z_L}\mu^{(a)}_i = (m_1,m_2,\ldots, m_a) \quad(a \in [1,n])
\end{align}
so that there are only finitely many choices for the summation variables.
Remark that (\ref{cst}) is equivalent to saying that 
$(\mu^{(a)}_1,\ldots, \mu^{(a)}_L)$ is a configuration 
of the $a$-iTAZRP in the sector $S(m_1,\ldots, m_a)$
which is again basic.
In particular from
$\sum_{i \in \Z_L} \mu^{(a),a}_i = m_a>0$,
there is at east one $i \in \Z_L$ such that $\mu^{(a),a}_i>0$,
therefore the associated operator
$A^{(a)}_{\mu^{(a-1)}_i, \mu^{(a)}_i}$
reduces to the `last term' involving $w_aK_a$ in (\ref{adef}).
Since it contains $w_a{\bf k}^{\otimes a-1}$,
the factor $\mathrm{Tr}_{F^{\otimes a-1}}(
A^{(a)}_{\mu^{(a-1)}_1, \mu^{(a)}_1}\cdots
A^{(a)}_{\mu^{(a-1)}_L, \mu^{(a)}_L})$ is finite 
and divisible by $w_a$ for $a \in [2,n]$.
Thus the homogeneous degree of 
$(w_2\cdots w_n)^{-1}
\mathrm{Tr}\bigl(X_{\sigma_1} \cdots X_{\sigma_L}\bigr)$
is $(n-1)L-(n-1) = (n-1)(L-1)$.
The fact that the coefficients belong to $\Z_{\ge 0}$ is obvious.
\end{proof}

The key ingredient for proving Theorem \ref{th:P} is  
\begin{proposition}[Generalized hat relation]\label{pr:hat}
Let $n \ge 2$.
For any $\alpha, \beta \in(\Z_{\ge 0})^n$ and 
$\mu, \nu \in(\Z_{\ge 0})^{n-1}$, 
the $A_{\mu, \alpha}$ and $\hat{A}_{\mu,\alpha}$ 
in (\ref{adef}) and (\ref{ahdef}) satisfy 
\begin{align}\label{ghat}
\sum_{\gamma, \delta \in(\Z_{\ge 0})^n}
h^{\alpha, \beta}_{\gamma,\delta}
A_{\mu, \gamma}A_{\nu, \delta}
- \sum_{\kappa,\lambda\in(\Z_{\ge 0})^{n-1}} 
\overline{h}^{\,\kappa,\lambda}_{\mu,\nu}
A_{\kappa,\alpha}A_{\lambda,\beta}
= \hat{A}_{\mu,\alpha}A_{\nu,\beta}
-A_{\mu,\alpha}\hat{A}_{\nu,\beta}, 
\end{align}
where $\overline{h}^{\,\kappa,\lambda}_{\mu,\nu}$ denotes the 
matrix element (\ref{hdef}) of the local Markov matrix 
for the $(n-1)$-$\mathrm{iTAZRP}$ involving $w_1,\ldots, w_{n-1}$.
\end{proposition}
We have proved (\ref{ghat}) by
separating it into equalities on the  
coefficients of the monomial $w_iw_jw_k$ case by case 
for each triple $(i,j,k)$ with $1 \le i \le j \le k \le n$.
The calculation is direct but quite lengthy hence omitted here.
The relation (\ref{ghat}) reads as
$(A \otimes A)h - \overline{h}(A\otimes A) 
= \hat{A} \otimes A - A \otimes \hat{A}$ in the matrix notation, and 
has the same form as the inhomogeneous $n$-TASEP case \cite{AM}.

\vspace{0.3cm}
{\em Proof of Theorem \ref{th:P}}.
In view of (\ref{htn}) we are to show the first line of the following 
for any ${\boldsymbol \sigma} =(\sigma_1,\ldots, \sigma_L) \in 
S(m_1,\ldots, m_n)$:
\begin{align*}
0&=\sum_{i \in \Z_L}\sum_{\sigma'_i, \sigma'_{i+1}}
\!\!h^{\sigma_i, \sigma_{i+1}}_{\sigma'_i, \sigma'_{i+1}}
\mathrm{Tr}\bigl(
X_{\sigma_1}\cdots
X_{\sigma'_i}X_{\sigma'_{i+1}}\cdots
X_{\sigma_L}\bigr) \\
&= \sum_{{\boldsymbol \mu} \in S(\overline{\bf m})}
\!\!\!\!\mathrm{Tr}(X^{(n-1)}_{\mu_1}\cdots X^{(n-1)}_{\mu_L})
\sum_{i \in \Z_L}\sum_{\sigma'_i, \sigma'_{i+1}}
\!\!h^{\sigma_i, \sigma_{i+1}}_{\sigma'_i, \sigma'_{i+1}}
\mathrm{Tr}\bigl(
A^{(n)}_{\mu_1, \sigma_1}\cdots 
A^{(n)}_{\mu_i, \sigma'_i}A^{(n)}_{\mu_{i+1}, \sigma'_{i+1}}\cdots
A^{(n)}_{\mu_L, \sigma_L}\bigr),
\end{align*}
To get the second line 
we have substituted (\ref{xrec}) and 
set ${\boldsymbol \mu}=(\mu_1,\ldots,\mu_L)$.
Using the remark after (\ref{cst}),
we have regarded $\overline{\bf m} = (m_1,\ldots, m_{n-1})$ 
as a label of the basic sector $S(\overline{\bf m})$ of the $(n-1)$-iTAZRP.
Now Proposition \ref{pr:hat} can be applied to the sum over 
$\sigma'_i, \sigma'_{i+1}$.
Then the contribution from the right hand side of (\ref{ghat}) is cancelled
under the sum $\sum_{i \in \Z_L}$ thanks to the cyclicity of the trace.
Thus the relation to be verified becomes
\begin{align*}
0=\sum_{{\boldsymbol \mu} \in S(\overline{\bf m})}
\mathrm{Tr}\bigl(
A^{(n)}_{\mu_1, \sigma_1}\cdots
A^{(n)}_{\mu_L, \sigma_L}\bigr)
\sum_{i \in \Z_L}
\sum_{\mu'_i, \mu'_{i+1}}
\!\!\overline{h}^{\,\mu_i, \mu_{i+1}}_{\,\mu'_i, \mu'_{i+1}}
\,\mathrm{Tr}(X^{(n-1)}_{\mu_1}\cdots
X^{(n-1)}_{\mu'_i}X^{(n-1)}_{\mu'_{i+1}}\cdots X^{(n-1)}_{\mu_L}),
\end{align*}
where we have replaced the summation variables 
$(\mu_i, \mu_{i+1})$ with $(\mu'_i, \mu'_{i+1})$. 
By induction on $n$, the proof reduces to the case $n=2$.
From the convention mentioned after (\ref{xrec}), 
it amounts to checking 
$\sum_{i \in \Z_L}
\sum_{\mu'_i, \mu'_{i+1}}
\!\!\overline{h}^{\,\mu_i, \mu_{i+1}}_{\,\mu'_i, \mu'_{i+1}}=0$
for the $1$-iTAZRP local Markov matrix $\overline{h}$.
This has been shown in the paragraph preceding Example \ref{ex:LL}.
\qed

From Theorem \ref{th:P} and (\ref{xrec})  we have 
\begin{corollary}\label{co:kna}
The steady state probabilities $\mathbb{P}(\sigma_1,\ldots, \sigma_L)$ 
of the $n$-$\mathrm{iTAZRP}$
in the basic sector $S({\bf m})$ is expressed as
\begin{align}\label{prec}
\mathbb{P}(\sigma_1,\ldots, \sigma_L) = w_n^{-1}\!\!
\sum_{(\mu_1,\ldots,\mu_L)\in S(\overline{\bf m})}
\overline{\mathbb P}(\mu_1,\ldots,\mu_L)
\mathrm{Tr}_{F^{\otimes n-1}}\bigl(
A_{\mu_1, \sigma_1}\cdots
A_{\mu_L, \sigma_L}\bigr),
\end{align}
where $A_{\mu_i, \sigma_i}=A^{(n)}_{\mu_i, \sigma_i}$, ${\bf m}=(m_1,\ldots, m_n),
\overline{\bf m}=(m_1,\ldots, m_{n-1})$ and 
$\overline{\mathbb P}$ stands for the steady state probability 
in the $(n-1)$-$\mathrm{iTAZRP}$.
\end{corollary}
Corollary \ref{co:kna} is a recursion relation 
between the polynomials $\mathbb{P} \in \Z[w_1,\ldots, w_n]$
and $\overline{\mathbb P} \in \Z[w_1,\ldots, w_{n-1}]$.

\begin{example}\label{ex:nn1}
Let us confirm 
\begin{align*}
\mathbb{P}(\emptyset,\emptyset,12) = w_1^2+w_1w_2+w_2^2,
\quad
\mathbb{P}(\emptyset,2,1) = w_2^2+w_1w_2,
\quad
\mathbb{P}(\emptyset,1,2) = w_2^2,
\end{align*} 
which is contained in the 2-iTAZRP state $|\xi_3(1,1)\rangle$ 
in Example \ref{ex:LL}.
Since these configurations 
$(\sigma_1,\sigma_2,\sigma_3)$ 
are in the sector $S(1,1)$, the matrix product takes the form  
$\sum_{\mu_1,\mu_2,\mu_3}\mathrm{Tr}(
A_{\mu_1,\sigma_1}
A_{\mu_2,\sigma_2}
A_{\mu_3,\sigma_3})$
with the sum obeying $\mu_1+\mu_2+\mu_3=1$\footnote{
We write $A_{(\mu),(\alpha^1,\alpha^2)}$ for example simply as
$A_{\mu,\alpha^1\alpha^2}$.}.
Switching to the multiplicity representation they read as
\begin{align*}
\mathbb{P}(00,00,11)  &= 
w_2^{-1}\mathrm{Tr}(A_{0,00}A_{0,00}A_{1,11})
+w_2^{-1}\mathrm{Tr}(A_{0,00}A_{1,00}A_{0,11})
+w_2^{-1}\mathrm{Tr}(A_{1,00}A_{0,00}A_{0,11}),\\
\mathbb{P}(00,01,10)  &= 
w_2^{-1}\mathrm{Tr}(A_{0,00}A_{0,01}A_{1,10})
+w_2^{-1}\mathrm{Tr}(A_{0,00}A_{1,01}A_{0,10})
+w_2^{-1}\mathrm{Tr}(A_{1,00}A_{0,01}A_{0,10}),\\
\mathbb{P}(00,10,01)  &= 
w_2^{-1}\mathrm{Tr}(A_{0,00}A_{0,10}A_{1,01})
+w_2^{-1}\mathrm{Tr}(A_{0,00}A_{1,10}A_{0,01})
+w_2^{-1}\mathrm{Tr}(A_{1,00}A_{0,10}A_{0,01}).
\end{align*}
The relevant operators are given by
\begin{alignat*}{4}
&A_{0,00} = w_1{\bf d}+w_2{\bf k},\;\; &
&A_{1,00} = {\bf a}^+(w_1{\bf d}+w_2{\bf k}),\;\;&
&A_{0,01} = w_2{\bf k},\;\;&
&A_{1,01} = w_2{\bf a}^+{\bf k},\\
&A_{0,10} = (w_1{\bf d}+w_2{\bf k}){\bf a}^- ,\;\; &
&A_{1,10} = {\bf a}^+(w_1{\bf d}+w_2{\bf k}){\bf a}^- ,\;\;&
&A_{0,11} =w_2{\bf k}\,{\bf a}^-,\;\; & 
&A_{1,11} = w_2{\bf a}^+{\bf k}\,{\bf a}^-.
\end{alignat*}
By a direct calculation the above traces are evaluated as   
\begin{align*}
w_2^{-1}\begin{pmatrix}
\mathrm{Tr}(A_{0,00}A_{0,00}A_{1,11}) &
\mathrm{Tr}(A_{0,00}A_{1,00}A_{0,11}) &
\mathrm{Tr}(A_{1,00}A_{0,00}A_{0,11}) \\
\mathrm{Tr}(A_{0,00}A_{0,01}A_{1,10})&
\mathrm{Tr}(A_{0,00}A_{1,01}A_{0,10})&
\mathrm{Tr}(A_{1,00}A_{0,01}A_{0,10})\\
\mathrm{Tr}(A_{0,00}A_{0,10}A_{1,01})&
\mathrm{Tr}(A_{0,00}A_{1,10}A_{0,01})&
\mathrm{Tr}(A_{1,00}A_{0,10}A_{0,01})
\end{pmatrix} = 
\begin{pmatrix}
w_1^2 & w_1w_2 & w_2^2\\
0 & w_1w_2 & w_2^2 \\
w_2^2 & 0 & 0
\end{pmatrix}
\end{align*} 
reproducing the sought result.
\end{example}

\begin{example}\label{ex:nn2}
Let us confirm $\mathbb{P}(\emptyset,123) = (w_1+w_2)(w_1+w_2+w_3)$,
which is contained in the 3-iTAZRP state $|\xi_3(1,1,1)\rangle$ 
in Example \ref{ex:LL}.
This time we invoke Corollary \ref{co:kna}  and consider the formula 
\begin{align*}
&w_3^{-1}\sum_{
\mu^1_1+\mu^1_2=\mu^2_1+\mu^2_2=1} \overline{\mathbb P}(\mu^1_1\mu^2_1,\mu^1_2\mu^2_2)
\mathrm{Tr}(A_{\mu^1_1\mu^2_1,000}A_{\mu^1_2\mu^2_2,111})\\
&=w_3^{-1}
\overline{\mathbb P}(11,00)
\mathrm{Tr}(A_{11,000}A_{00,111})+
w_3^{-1}\overline{\mathbb P}(10,01)
\mathrm{Tr}(A_{10,000}A_{01,111})\\
&+
w_3^{-1}\overline{\mathbb P}(01,10)
\mathrm{Tr}(A_{01,000}A_{10,111})+
w_3^{-1}\overline{\mathbb P}(00,11)
\mathrm{Tr}(A_{00,000}A_{11,111}).
\end{align*}
For the 2-iTAZRP steady state probabilities $\overline{\mathbb P}$, we 
apply the result in Example \ref{ex:LL} and the cyclic symmetry to find 
$\overline{\mathbb P}(11,00)=\overline{\mathbb P}(00,11)=w_1+w_2$ and 
$\overline{\mathbb P}(10,01) = \overline{\mathbb P}(10,01) = w_2$.
The relevant operators are given by 
\begin{align*}
A_{00,000} &= w_1{\bf d} \otimes {\bf 1} + w_2 {\bf k} \otimes {\bf d} 
+ w_3 {\bf k} \otimes {\bf k},
\quad
A_{00,111} = w_3({\bf k} \otimes {\bf k})({\bf a}^- \otimes {\bf a}^-),\\
A_{10,000} &= ({\bf a}^+ \otimes {\bf 1})A_{00,000},\;\;
A_{01,000} = ({\bf 1} \otimes {\bf a}^+)A_{00,000},\;\;
A_{11,000} = ({\bf a}^+ \otimes {\bf a}^+)A_{00,000},\\
A_{10,111}  &= ({\bf a}^+ \otimes {\bf 1})A_{00,111},\;\;
A_{01,111}  = ({\bf 1} \otimes {\bf a}^+)A_{00,111},\;\;
A_{11,111}  = ({\bf a}^+ \otimes {\bf a}^+)A_{00,111}.
\end{align*}
Then the above traces are evaluated as
\begin{align*}
w_3^{-1}\begin{pmatrix}
\mathrm{Tr}(A_{11,000}A_{00,111}) &
\mathrm{Tr}(A_{10,000}A_{01,111}) \\
\mathrm{Tr}(A_{01,000}A_{10,111}) &
\mathrm{Tr}(A_{00,000}A_{11,111})
\end{pmatrix} = 
\begin{pmatrix}
w_3 & w_2\\
w_1 & w_1
\end{pmatrix}
\end{align*} 
reproducing the sought result.

\end{example}

\section{Combinatorial construction of steady state}\label{sec:mp}

Steady states of the $n$-iTAZRP can also be constructed via 
a combinatorial algorithm generalizing the one in \cite[Sec.4.3]{KMO3} 
by introducing the inhomogeneity parameters $w_1,\ldots, w_n$\footnote{
The convention of labeling the particle species 
here is opposite from \cite{KMO3} causing many changes.}.
It may be viewed as a TAZRP analogue of the results
on the multispecies inhomogeneous TASEP \cite{AL,AM} whose 
homogeneous case $w_1= \cdots =w_n$ goes back to \cite{FM}.  

\subsection{Combinatorial formula}

Given a multiplicity array 
${\bf m}=(m_1,\ldots, m_n) \in (\Z_{\ge 1})^n$ specifying 
a basic sector $S({\bf m})$ of the $n$-iTAZRP 
in the periodic chain $\Z_L$, define 
$\ell_1, \ldots, \ell_n$ by 
\begin{align}\label{mkr:akci}
\ell_a = m_1+m_2 + \cdots + m_a
\end{align}
so that $1\le \ell_1< \ell_2 < \cdots < \ell_n$. 
Associated with the data we introduce the finite sets
\begin{equation}\label{mkr}
\begin{split}
B({\bf m}) &= B_{\ell_n} \otimes \cdots \otimes B_{\ell_1}
= \{{\bf x}= {\bf x}^n\otimes \cdots \otimes {\bf x}^1 \mid
{\bf x}^a = (x^a_1,\ldots, x^a_L) \in B_{\ell_a}\},\\
B_\ell &=\{(x_1,\ldots, x_L)\in (\Z_{\ge 0})^L\mid
x_1+ \cdots + x_L = \ell\},
\end{split}
\end{equation} 
where $\otimes$ may just be regarded as the direct product of sets.
Elements in $B({\bf m})$ were called {\em multiline states} in \cite{KMO3}
in analogy with \cite{FM}.
The $B({\bf m})$ is endowed with the structure of a crystal \cite{Ka1} of
the the quantum affine algebra $U_q(\widehat{sl}_L)$ \cite{D86,J},
although this aspects will not be used in the sequel.
Our main task, which will be detailed in the 
next subsection, is to formulate a projection $\pi$ and 
a {\em weight} function $W$ 
\begin{align}\label{hzks}
\pi: B({\bf m}) \rightarrow S({\bf m}),\qquad
W: B({\bf m}) \rightarrow 
\{ w_1^{p_1}\cdots w_n^{p_n} \mid p_1,\ldots, p_n \in \Z_{\ge 0}\}
\end{align}
such that the following formula holds:
\begin{theorem}\label{th:we}
The steady state probability 
in Theorem \ref{th:P} in the sector $S({\bf m})$ is expressed as
\begin{align*}
\mathbb{P}({\boldsymbol \sigma})
= \sum_{{\bf x} \in \pi^{-1}({\boldsymbol \sigma})}W({\bf x}),
\end{align*}
or equivalently, the steady state is constructed as
\begin{align*}
|\mathscr{P}_L({\bf m})\rangle
= \sum_{{\bf x} \in B({\bf m})}W({\bf x})|\pi({\bf x})\rangle.
\end{align*}
\end{theorem}
The proof will be given after Remark \ref{re:okne}.

\subsection{Construction of the maps $\pi$ and $W$}

The map $\pi$ has been constructed in \cite[Sec.4]{KMO3}.
Let us recall it in the form adapted to the present convention.  
The weight function $W$, which is new, will also be determined in the course of it.
Our construction is recursive with respect to $n$.

Let  $a \in [2,n]$.
Associated with any ${\bf x}^a \in B_{\ell_a}$ 
we introduce the maps\footnote{The notation
$\Phi_{{\bf x}^a}$ is slightly incomplete in that 
the $a$-dependence becomes invisible when ${\bf x}^a$ is written generally as 
${\bf y}$ for example. However we prefer it for simplicity. A proper alternative is 
to formulate it as a map $S(m_1,\ldots, m_{a-1})\times B_{\ell_a}\;
\longrightarrow S(m_1,\ldots, m_a)$. 
A similar caution applies to $\varpi_{{\bf x}^a}$.}
\begin{equation}\label{utkr}
\begin{split}
\Phi_{{\bf x}^a}: &S(m_1,\ldots, m_{a-1})
\rightarrow S(m_1,\ldots, m_a),\\
\varpi_{{\bf x}^a}: &S(m_1,\ldots, m_{a-1})
\rightarrow \{ w_1^{p_1}\cdots w_a^{p_a} \mid p_1,\ldots, p_a \in \Z_{\ge 0}\}.
\end{split}
\end{equation}
Note that $S(m_1,\ldots, m_{a-1})$ and 
$S(m_1,\ldots, m_a)$ are the sets of configurations of  
$(a-1)$-iTAZRP and $a$-iTAZRP 
in the basic sectors.
The map $\Phi_{{\bf x}^{a}}$ 
called TAZRP {\em embedding rule} in \cite{KMO3}
and $\varpi_{{\bf x}^a}$ 
are defined through {\it Step} 0 -- {\it Step} 3 in the sequel.

\vspace{0.2cm}
{\it Step} 0.  Draw the $(a-1)$-iTAZRP configuration 
${\boldsymbol \sigma}$ above and the 
${\bf x}^a\in B_{\ell_a}$ below
in the two-row diagram as the following example for  
$L=7, a=4, \ell_a=9, {\bf x}^4=(0,2,1,2,0,1,3)\in B_9$
and the 3-iTAZRP state 
${\boldsymbol \sigma}=(\emptyset,13,2,3,\emptyset,12,11)$ 
in multiset representation\footnote{
This is an abbreviation of 
$(\emptyset, \{1,3\}, \{2\}, \{3\}, \emptyset, \{1,2\}, \{1,1\})$ as in 
Example \ref{ex:LL}.}.
\begin{equation*}
\begin{picture}(300,55)(-65,-1)

\multiput(0,0)(25,0){8}{\put(0,0){\line(0,1){50}}}
\multiput(0,0)(0,25){3}{\put(0,0){\line(1,0){175}}}

\put(-100,33){3-iTAZRP state \;${\boldsymbol \sigma}$}
\put(-25,9){${\bf x}^4$}
\put(35,39){3}\put(35,28){1}
\put(60,34){2}
\put(85,34){3}
\put(135,39){2}\put(135,28){1}
\put(160,39){1}\put(160,28){1}
\put(35,14){$\bullet$}\put(35,4){$\bullet$}
\put(60,9){$\bullet$}
\put(85,14){$\bullet$}\put(85,4){$\bullet$}
\put(135,9){$\bullet$}
\put(160,16.5){$\bullet$}\put(160,9.5){$\bullet$}\put(160,2.5){$\bullet$}

\end{picture}
\end{equation*}
The element ${\bf x}^a\in B_{\ell_a}$ is depicted as a dot pattern.
The positions of the dots and the particles (numbers in the top row) 
{\em within} a box do not matter.
Each dot is either {\em colored} or {\em uncolored}.
Initially they are all uncolored.

\vspace{0.2cm}
{\it Step} 1. Do the pairing procedures $(1), (2), \ldots, (a\!-\!1)$ in this order, where
$(b)$ is to connect the species $b$ particles upstairs
to partner dots downstairs according to 
$(b-\mathrm{i})$ and 
$(b-\mathrm{ii})$ below,
where  ``W"(west) and ``left" are meant in the periodic sense,
namely, W or left to the leftmost column is the rightmost one.

\begin{itemize}
\item[$(b-\mathrm{i})$]
Pick any species $b$ particle\footnote{The arbitrariness of the choice  
does not spoil the well-definedness. See Remark \ref{re:okne} (i).} 
and draw a line from it to its SW neighbor box
via the {\em left adjacent} box\footnote{Hence every line starts 
from a particle with the shape 
$\lceil$ rather than $\rfloor$.}.  
If there are uncolored dots in the SW neighbor box, 
any one of them can be chosen as the partner. 
If uncolored dots are absent there,  the line should further proceed to the left until it 
first encounters an uncolored dot. 
If there are more than one uncolored dots in the same box,
any one of them can be made the partner.
Treat the so captured partner dot as colored onward. 
The line connecting $b$ and its partner shall be called an {\em H-line} with color $b$.

\item[$(b-\mathrm{ii})$]
Repeat $(b-\mathrm{i})$ for all the remaining species $b$ particles upstairs.
\end{itemize}

\vspace{0.2cm}
{\it Step} 2. Regard the colored dots captured 
by color $b$ $H$-lines as species $b$ particles
for $b \in [1,a-1]$.
Regard yet uncolored dots as species $a$ particles.
Then the bottom line gives the 
$a$-iTAZRP configuration $\Phi_{{\bf x}^{a}}({\boldsymbol \sigma})$.

\vspace{0.2cm}
{\it Step} 3. From the diagram obtained by completing {\it Step} 1, 
$\varpi_{{\bf x}^a}({\boldsymbol \sigma})$ is determined as 
\begin{align}\label{eta}
&\varpi_{{\bf x}^a}({\boldsymbol \sigma}) =
w_a^{-1}\eta_1\eta_2 \cdots \eta_L,
\qquad
\eta_i= \begin{cases}
w_a & \text{if}\;\; \mathrm{col}_i = \emptyset, \\
w_{\min(\mathrm{col}_i)} & \text{otherwise},
\end{cases}
\end{align}
where 
$\mathrm{col}_i\,(i \in \Z_L)$ is the set of colors of the $H$-lines that cross 
the border of the $i$-th and the $(i+1)$-th boxes in the {\em bottom} row horizontally.

\vspace{0.2cm}
In our ongoing example, doing $(1)$ to the previous diagram leads to
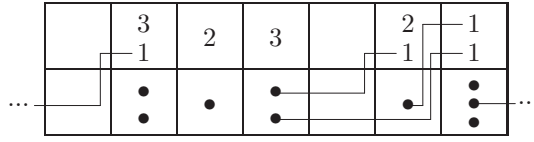
\begin{equation*}
\begin{picture}(300,55)(-65,-1)

\multiput(0,0)(25,0){8}{\put(0,0){\line(0,1){50}}}
\multiput(0,0)(0,25){3}{\put(0,0){\line(1,0){175}}}

\put(35,39){3}\put(35,28){1}
\put(60,34){2}
\put(85,34){3}
\put(135,39){2}\put(135,28){1}
\put(160,39){1}\put(160,28){1}
\put(35,14){$\bullet$}\put(35,4){$\bullet$}
\put(60,9){$\bullet$}
\put(85,14){$\bullet$}\put(85,4){$\bullet$}
\put(135,9){$\bullet$}
\put(160,16.5){$\bullet$}\put(160,9.5){$\bullet$}\put(160,2.5){$\bullet$}

\drawline(158,42)(143,42)
\drawline(143,42)(143,11.3)\drawline(143,11.3)(136,11.3)

\drawline(158,31)(146,31)
\drawline(146,31)(146,6)\drawline(146,6)(86,6)

\drawline(133,31)(121,31)\drawline(121,31)(121,16)
\drawline(121,16)(86,16)

\drawline(33,31)(21,31)\drawline(21,31)(21,11.3)
\drawline(21,11.3)(-4,11.3)\put(-14,11){...}
\put(0,0.7){\drawline(163,11.3)(178,11.3)\put(179,11){...}}

\end{picture}
\end{equation*}
Doing $(2)$ to this diagram leads to
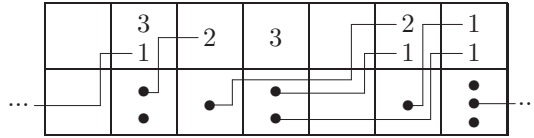
\begin{equation*}
\begin{picture}(300,55)(-65,-1)

\multiput(0,0)(25,0){8}{\put(0,0){\line(0,1){50}}}
\multiput(0,0)(0,25){3}{\put(0,0){\line(1,0){175}}}

\put(35,39){3}\put(35,28){1}
\put(60,34){2}
\put(85,34){3}
\put(135,39){2}\put(135,28){1}
\put(160,39){1}\put(160,28){1}
\put(35,14){$\bullet$}\put(35,4){$\bullet$}
\put(60,9){$\bullet$}
\put(85,14){$\bullet$}\put(85,4){$\bullet$}
\put(135,9){$\bullet$}
\put(160,16.5){$\bullet$}\put(160,9.5){$\bullet$}\put(160,2.5){$\bullet$}

\drawline(158,42)(143,42)
\drawline(143,42)(143,11.3)\drawline(143,11.3)(136,11.3)

\drawline(158,31)(146,31)
\drawline(146,31)(146,6)\drawline(146,6)(86,6)

\drawline(133,42)(116,42)
\drawline(116,42)(116,21)\drawline(116,21)(71,21)
\drawline(71,21)(71,11.5)\drawline(71,11.5)(62.5,11.5)

\drawline(133,31)(121,31)\drawline(121,31)(121,16)
\drawline(121,16)(86,16)

\drawline(58,37)(43,37)
\drawline(43,37)(43,16.2)\drawline(43,16.2)(36,16.2)

\drawline(33,31)(21,31)\drawline(21,31)(21,11.3)
\drawline(21,11.3)(-4,11.3)\put(-14,11){...}
\put(0,0.7){\drawline(163,11.3)(178,11.3)\put(179,11){...}}

\end{picture}
\end{equation*}
One of the color 2 $H$-line is making extra 90$^\circ$ turns just 
before capturing the partner dot rather than going straight.
This is just by the presentational reason to avoid the intersection of $H$-lines and 
has no significance. 
Doing $(3)$ similarly to this diagram leads to
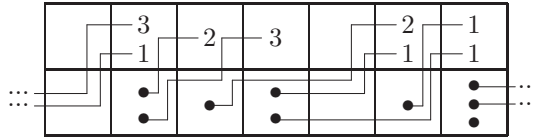
\begin{equation*}
\begin{picture}(300,55)(-65,-1)

\multiput(0,0)(25,0){8}{\put(0,0){\line(0,1){50}}}
\multiput(0,0)(0,25){3}{\put(0,0){\line(1,0){175}}}

\put(35,39){3}\put(35,28){1}
\put(60,34){2}
\put(85,34){3}
\put(135,39){2}\put(135,28){1}
\put(160,39){1}\put(160,28){1}
\put(35,14){$\bullet$}\put(35,4){$\bullet$}
\put(60,9){$\bullet$}
\put(85,14){$\bullet$}\put(85,4){$\bullet$}
\put(135,9){$\bullet$}
\put(160,16.5){$\bullet$}\put(160,9.5){$\bullet$}\put(160,2.5){$\bullet$}

\drawline(158,42)(143,42)
\drawline(143,42)(143,11.3)\drawline(143,11.3)(136,11.3)

\drawline(158,31)(146,31)
\drawline(146,31)(146,6)\drawline(146,6)(86,6)

\drawline(133,42)(116,42)
\drawline(116,42)(116,21)\drawline(116,21)(71,21)
\drawline(71,21)(71,11.5)\drawline(71,11.5)(62.5,11.5)

\drawline(133,31)(121,31)\drawline(121,31)(121,16)
\drawline(121,16)(86,16)

\drawline(58,37)(43,37)
\drawline(43,37)(43,16.2)\drawline(43,16.2)(36,16.2)

\drawline(82,37)(67,37)\drawline(67,37)(67,21)
\drawline(67,21)(46.5,21)\drawline(46.5,21)(46.5,6)
\drawline(46.5,6)(36,6)

\drawline(33,31)(21,31)\drawline(21,31)(21,11.3)
\drawline(21,11.3)(-4,11.3)\put(-14,11){...}
\put(0,0.7){\drawline(163,11.3)(178,11.3)\put(179,11){...}}

\drawline(33,42)(16,42)\drawline(16,42)(16,16)
\drawline(16,16)(-4,16)\put(-14,15.7){...}
\put(0,0.2){\drawline(163,18.3)(178,18.3)\put(179,18){...}}

\end{picture}
\end{equation*}
{\it Step} 1 is completed.
In the next diagram 
we display the 4-iTAZRP state 
$\Phi_{{\bf x}^a}({\boldsymbol \sigma})$ obtained as 
the result of {\it Step} 2 together with 
the quantities $\eta_i$'s in (\ref{eta}).
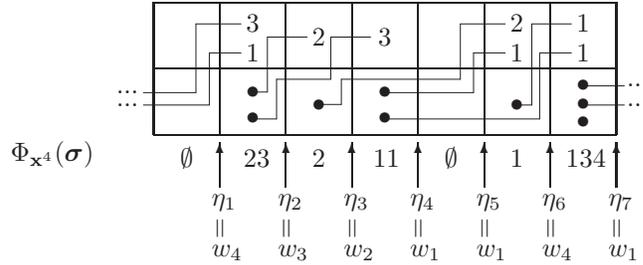
\begin{equation*}
\begin{picture}(300,105)(-65,-50)

\multiput(0,0)(25,0){8}{\put(0,0){\line(0,1){50}}}
\multiput(0,0)(0,25){3}{\put(0,0){\line(1,0){175}}}

\put(35,39){3}\put(35,28){1}
\put(60,34){2}
\put(85,34){3}
\put(135,39){2}\put(135,28){1}
\put(160,39){1}\put(160,28){1}
\put(35,14){$\bullet$}\put(35,4){$\bullet$}
\put(60,9){$\bullet$}
\put(85,14){$\bullet$}\put(85,4){$\bullet$}
\put(135,9){$\bullet$}
\put(160,16.5){$\bullet$}\put(160,9.5){$\bullet$}\put(160,2.5){$\bullet$}

\drawline(158,42)(143,42)
\drawline(143,42)(143,11.3)\drawline(143,11.3)(136,11.3)

\drawline(158,31)(146,31)
\drawline(146,31)(146,6)\drawline(146,6)(86,6)

\drawline(133,42)(116,42)
\drawline(116,42)(116,21)\drawline(116,21)(71,21)
\drawline(71,21)(71,11.5)\drawline(71,11.5)(62.5,11.5)

\drawline(133,31)(121,31)\drawline(121,31)(121,16)
\drawline(121,16)(86,16)

\drawline(58,37)(43,37)
\drawline(43,37)(43,16.2)\drawline(43,16.2)(36,16.2)

\drawline(82,37)(67,37)\drawline(67,37)(67,21)
\drawline(67,21)(46.5,21)\drawline(46.5,21)(46.5,6)
\drawline(46.5,6)(36,6)

\drawline(33,31)(21,31)\drawline(21,31)(21,11.3)
\drawline(21,11.3)(-4,11.3)\put(-14,11){...}
\drawline(163,11.3)(178,11.3)\put(179,11){...}

\drawline(33,42)(16,42)\drawline(16,42)(16,16)
\drawline(16,16)(-4,16)\put(-14,15.7){...}
\put(0,0.2){\drawline(163,18.3)(178,18.3)\put(179,18){...}}

\put(-54,-10){$\Phi_{{\bf x}^4}({\boldsymbol \sigma})$}
\put(10,-12){$\emptyset$}
\put(34,-12){$2 3$}
\put(60,-12){$2$}
\put(83,-12){$1 1$}
\put(110,-12){$\emptyset$}
\put(135,-12){$1$}
\put(156,-12){$134$}

\put(25,-20){\vector(0,1){17}}
\put(22,-27){$\eta_1$}\put(24,-37)
{\rotatebox[origin=c]{90}{$=$}}\put(22,-46){$w_4$}

\put(25,0){
\put(25,-20){\vector(0,1){17}}
\put(22,-27){$\eta_2$}\put(24,-37)
{\rotatebox[origin=c]{90}{$=$}}\put(22,-46){$w_3$}}

\put(50,0){
\put(25,-20){\vector(0,1){17}}
\put(22,-27){$\eta_3$}\put(24,-37)
{\rotatebox[origin=c]{90}{$=$}}\put(22,-46){$w_2$}}

\put(75,0){
\put(25,-20){\vector(0,1){17}}
\put(22,-27){$\eta_4$}\put(24,-37)
{\rotatebox[origin=c]{90}{$=$}}\put(22,-46){$w_1$}}

\put(100,0){
\put(25,-20){\vector(0,1){17}}
\put(22,-27){$\eta_5$}\put(24,-37)
{\rotatebox[origin=c]{90}{$=$}}\put(22,-46){$w_1$}}

\put(125,0){
\put(25,-20){\vector(0,1){17}}
\put(22,-27){$\eta_6$}\put(24,-37)
{\rotatebox[origin=c]{90}{$=$}}\put(22,-46){$w_4$}}

\put(150,0){
\put(25,-20){\vector(0,1){17}}
\put(22,-27){$\eta_7$}\put(24,-37)
{\rotatebox[origin=c]{90}{$=$}}\put(22,-46){$w_1$}}

\end{picture}
\end{equation*}
In this way we obtain the 4-iTAZRP configuration
$\Phi_{{\bf x}^4}({\boldsymbol \sigma}) 
= (\emptyset, 23, 2,11,\emptyset,1,134)$ 
in multiset representation and 
$\varpi_{{\bf x}^4}({\boldsymbol \sigma})= 
w_1^3w_2w_3w_4$.

Having formulated the maps (\ref{utkr}), 
we can now define $\pi$ and $W$ in (\ref{hzks}) as 
\begin{align}\label{mddm}
\pi({\bf x}^n \otimes \cdots \otimes {\bf x}^1) = {\boldsymbol \sigma}^n,\qquad
W({\bf x}^n \otimes \cdots \otimes {\bf x}^1) = \prod_{a=2}^n
\varpi_{{\bf x}^a}({\boldsymbol \sigma}^{a-1}).
\end{align} 
Here ${\boldsymbol \sigma}^a$ is the 
$a$-iTAZRP configuration in the basic sector 
$S(m_1,\ldots, m_a)$ constructed as 
\begin{align}\label{mgdn}
{\boldsymbol \sigma}^a = \Phi_{{\bf x}^a}({\boldsymbol \sigma}^{a-1}) =
\Phi_{{\bf x}^a}\circ \Phi_{{\bf x}^{a-1}}
\circ \cdots \circ \Phi_{{\bf x}^2}({\bf x}^1)\quad (a \in [2,n])
\end{align}
by identifying ${\bf x}^1=(x^1_1,\ldots, x^1_L) \in B_{\ell_1=m_1}$ 
with the 1-iTAZRP configuration
$(x^1_1,\ldots, x^1_L)$ in multiplicity representation, or equivalently 
$(1^{x^1_1},\ldots, 1^{x^1_L})$ in multiset representation.
It is in the basic sector $S(m_1)$.
In particular when $n=1$ we understand these definitions as
$\pi({\bf x}^1)= {\boldsymbol \sigma}^1 = {\bf x}^1$ and 
$W({\bf x}^1)=1$.

\begin{remark}\label{re:okne}
$ $
\begin{enumerate}

\item {\it Step} 1 is a part of the 
algorithm for the combinatorial $R$ for $U_q(\widehat{sl}_L)$ known as 
the NY-rule \cite[Rule 3.11]{NY}. 
In $(b-\mathrm{i})$ and $(b-\mathrm{ii})$,
the $H$-lines depend on the order of choosing species $b$ particles from  
${\boldsymbol \sigma}$. 
However the resulting 
$\Phi_{{\bf x}^a}({\boldsymbol \sigma})$ and 
$\varpi_{{\bf x}^a}({\boldsymbol \sigma})$ are independent of it 
due to \cite[Prop.3.20]{NY}. 

\item In (\ref{eta}), $\eta_1\cdots \eta_L$ is always divisible by $w_a$.
To see this, notice that there remains exactly $m_a\ge 1$ uncolored dots 
after {\it Step} 1 because the numbers of the particles upstairs is 
$\ell_{a-1}$ and the dots downstairs is $\ell_a=\ell_{a-1}+m_a$.
Thus there is at least one box in the bottom row containing at least one uncolored dot.
The left border of such a box must not be crossed by an $H$-line 
as it contradicts the rule under which they are drawn.  
Thus the factor $\eta_i$ assigned to this border is $w_a$.
In the above example, $\eta_6$ is it. 
(There can be more than one such $\eta_i$ in general.)
\end{enumerate}
\end{remark}

We have completed the description of the maps $\pi$ and $W$.
The nested construction (\ref{mgdn}) based on the combinatorial $R$ 
resembles the combinatorial Bethe ansatz \cite{KOSTY}.
However it is not known to us if $W$ has a natural interpretation 
in crystal theory like {\em energy} \cite{NY}.

\vspace{0.2cm}
{\em Proof of Theorem \ref{th:we}}.
The claim is actually the combinatorial interpretation 
of the matrix product formula in Theorem \ref{th:P}.
We explain it along the factor 
$A^{(a)}_{\mu^{(a-1)}_i, \mu^{(a)}_i}$ in (\ref{ryb}).
For simplicity write it as
$A^{(a)}_{(\mu^1,\ldots, \mu^{a-1}), (\alpha^1,\ldots, \alpha^a)}$ 
and consider its summand involving $K_r$ with $r \in [1,a-1]$ given in 
(\ref{adef}):
\begin{align}
&w_r\delta_{\alpha^{r+1}, \ldots,  \alpha^a}^{\; \,0,\, . . . \ldots, \, 0}
P_+(\mu)
(\overbrace{{\bf k} \otimes \cdots \otimes {\bf k}}^{r-1}\otimes 
{\bf d} \otimes \overbrace{{\bf 1} \otimes \cdots \otimes {\bf 1}}^{a-1-r})
P_-(\overline{\alpha}),
\label{bkkj}\\
&P_+(\mu) = 
({\bf a}^+)^{\mu^1} \otimes \cdots \otimes ({\bf a}^+)^{\mu^{a-1}},\quad
P_-(\overline{\alpha}) = 
({\bf a}^-)^{\alpha^1} \otimes \cdots \otimes ({\bf a}^-)^{\alpha^r}\otimes 
\overbrace{{\bf 1} \otimes \cdots \otimes {\bf 1}}^{a-1-r},
\nonumber
\end{align}
where we have taken the condition $\alpha^{r+1}= \cdots = \alpha^a=0$ 
in $P_-(\overline{\alpha}) $ into account. 
The operator (\ref{bkkj})  acts on the vectors of the form 
$|s_1\rangle \otimes \cdots \otimes  |s_{a-1}\rangle \in F^{\otimes a-1}$.
We interpret it as the local situation in the two-row diagram obtained by {\em Step 1} 
in which there are $s_b$ $H$-lines with color $b$
coming into the $i$-th bottom box from the left or from the above 
as follows\footnote{Possible incoming $H$-lines originating from the 
NW neighbor box have not been drawn here for simplicity 
but are included in the argument.}: 
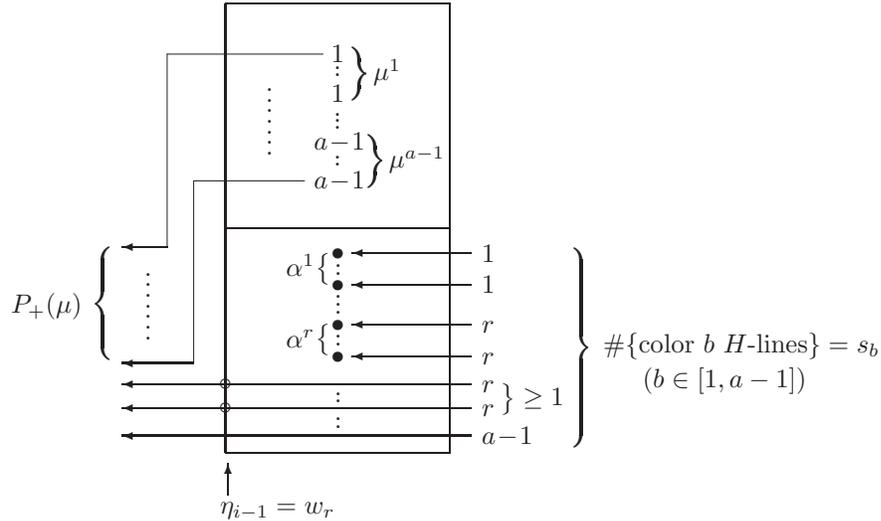
\begin{equation*}
\begin{picture}(200,200)(-55,-58)
\multiput(0,0)(0,85){2}{
\put(-18,-33){
\multiput(0,0)(85,0){2}{\put(0,0){\line(0,1){85}}}
\multiput(0,0)(0,85){2}{\put(0,0){\line(1,0){85}}}}}
\put(0,25){
\put(22,90){1}
\multiput(0,-1)(0,2.5){2}{\put(23.1,85.7){.}}
\put(22,75){1}
\put(29,82){\rotatebox[origin=c]{-90}
{$\overbrace{\phantom{\quad\quad}}$}}
\put(38,83){$\mu^1$}
}

\multiput(0,4)(0,2.5){3}{\put(23.1,85.7){.}}

\put(0,-8){
\put(15.5,90){$a\!-\!1$}
\multiput(0,-2)(0,2.5){2}{\put(23.1,85.7){.}}
\put(15.5,75){$a\!-\!1$}
\put(35,82){\rotatebox[origin=c]{-90}
{$\overbrace{\phantom{\quad\quad}}$}}
\put(44,83){$\mu^{a-1}$}
}

\drawline(19,118)(-40,118)\drawline(-40,118)(-40,45)
\put(-40,45){\vector(-1,0){17}}

\drawline(12,70)(-30,70)\drawline(-30,70)(-30,1)
\put(-30,1){\vector(-1,0){27}}


\put(22,40){$\bullet$}\put(75,42.5){\vector(-1,0){45}}\put(79,39.5){$1$}
\multiput(0,4.5)(0,3){2}{\put(23.4,30){.}}
\put(22,28){$\bullet$}\put(75,30.5){\vector(-1,0){45}}\put(79,27.5){$1$}
\put(5,34){$\alpha^1$}

\multiput(0,-10.4)(0,3){3}{\put(23.4,30){.}}
\multiput(-26,13)(0,4){7}{\put(23.4,67){.}}

\put(0,-27){
\put(22,40){$\bullet$}\put(75,42.5){\vector(-1,0){45}}\put(79,39.5){$r$}
\multiput(0,4.5)(0,3){2}{\put(23.4,30){.}}
\put(22,28){$\bullet$}\put(75,30.5){\vector(-1,0){45}}\put(79,27.5){$r$}
\put(5,34){$\alpha^r$}
}

\put(75,-7){\vector(-1,0){132}}\put(79,-10){$r$}\put(-20.5,-9.3){$\circ$}
\put(75,-16){\vector(-1,0){132}}\put(79,-19){$r$}\put(-20.5,-18.3){$\circ$}
\put(86,-15){$\bigr\} \ge 1$}

\multiput(-72,-57)(0,4){7}{\put(23.4,67){.}}
\put(-68,20){\rotatebox[origin=c]{90}
{$\overbrace{\phantom{\qquad\quad\quad\;}}$}}
\put(-99,20){$P_+(\mu)$}

\multiput(0,-43.7)(0,2.7){2}{\put(23.4,30){.}}
\multiput(0,-53.5)(0,3){2}{\put(23.4,30){.}}

\put(0,1.5){
\put(75,-28){\vector(-1,0){132}}\put(79,-31){$a\!-\!1$}
}

\put(-17,-49){\vector(0,1){12}}
\put(-20,-56){$\eta_{i-1} = w_r$}

\put(16,34){$\bigl\{$}
\put(16,7.6){$\bigl\{$}

\put(113,5){\rotatebox[origin=c]{-90}{
$\overbrace{\phantom{\qquad\qquad\qquad\quad\;\;}}$}}
\put(125,5){$\#\{\text{color $b$ $H$-lines}\}=s_b$}
\put(140,-8){$(b \in [1,a-1])$}

\end{picture}
\end{equation*}
The ${\bf a}^+$ and ${\bf a}^-$ at the $b$-th tensor component
of $P_+(\mu)$ and $P_-(\overline{\alpha})$
work as the emission and absorption of $H$-lines of color $b$, respectively.
On the other hand ${\bf k}$ compels the complete absorption allowing no $H$-lines 
to penetrate the left border while ${\bf d}={\bf 1}-{\bf k}$ 
demands at least one $H$-line with color $r$ does penetrate leading to $\eta_{i-1}=w_r$
as marked by $\circ$ in the above diagram.
The term involving $w_aK_a$ similarly selects the situation in which all the $H$-lines 
are absorbed and assigns it with the factor $\eta_{i-1}=w_a$.
\qed

\begin{example}\label{ex:knd}
Let us consider 4-iTAZRP on the length $L=4$ chain and 
the configuration ${\boldsymbol \sigma}=(3, 14, \emptyset, 22)$ in multiset representation.
It is in the basic sector $S(1,2,1,1)$.
The steady state probability is given by  
$\mathbb{P}({\boldsymbol \sigma}) = w_2^2 w_3^2 w_4^2
(w_1w_2w_4+w_2w_3w_4+w_1w_3w_4+w_1w_2w_3)$.
The relevant set (\ref{mkr}) is 
$B({\bf m}) = B_5 \otimes B_4 \otimes B_3 \otimes B_1$.
We have $\pi^{-1}({\boldsymbol \sigma})  
= \{{\bf y}_1, {\bf y}_2, {\bf y}_3, {\bf y}_4\}$, where
${\bf y}_j \in  B({\bf m})$ reads
\begin{align*}
{\bf y}_1 &= (1, 2, 0, 2) \otimes (2, 1, 1, 0) \otimes (1, 2, 0, 0) \otimes (0, 1, 0, 0),\\
{\bf y}_2 &= (1, 2, 0, 2) \otimes (2, 1, 1, 0) \otimes (0, 2, 0, 1) \otimes (1, 0, 0, 0),\\
{\bf y}_3 &= (1, 2, 0, 2) \otimes (2, 1, 1, 0) \otimes (0, 2, 0, 1) \otimes (0, 1, 0, 0),\\
{\bf y}_4 &= (1, 2, 0, 2) \otimes (2, 1, 0, 1) \otimes (1, 2, 0, 0) \otimes (0, 1, 0, 0).
\end{align*}
Then the data (\ref{eta}), (\ref{mddm}) and (\ref{mgdn}) are determined as follows.
\begin{center}
\begin{tabular}{c|ccccc}
${\bf y}$  &  ${\bf y}_1$ & ${\bf y}_2$ &  ${\bf y}_3$ & ${\bf y}_4$ \\
\hline
$\phantom{\overbrace{A}}\!\!\!\!\!\!\!\!\!\!{\boldsymbol \sigma}^1$ 
& $(\emptyset,1,\emptyset, \emptyset)$
& $(1, \emptyset,\emptyset,\emptyset)$
& $(\emptyset,1,\emptyset,\emptyset)$ 
& $(\emptyset,1,\emptyset,\emptyset)$ \\

${\boldsymbol \sigma}^2$ 
& $(1,22,\emptyset, \emptyset)$
& $(\emptyset,22,\emptyset,1)$
& $(\emptyset,22,\emptyset,1)$ 
& $(1,22,\emptyset,\emptyset)$ \\

${\boldsymbol \sigma}^3$ 
& $(22, 3,1, \emptyset)$
& $(22, 3,1,\emptyset)$
& $(22, 3,1,\emptyset)$ 
& $(22, 3, \emptyset,1)$ \\

$\varpi_{{\bf y}^2}({\boldsymbol \sigma}^1)$ 
& $w_2^3$ & $w_2^3$ & $w_1w_2^2$ & $w_2^3$ \\ 

$\varpi_{{\bf y}^3}({\boldsymbol \sigma}^2)$ 
& $w_1w_3^2$ & $w_3^3$ & $w_3^3$ & $w_3^3$ \\ 

$\varpi_{{\bf y}^4}({\boldsymbol \sigma}^3)$ 
& $w_4^3$ & $w_4^3$ & $w_4^3$ & $w_1w_4^2$ \\

$W({\bf y})$ 
& $w_1w_2^3w_3^2w_4^3$ 
& $w_2^3w_3^3w_4^3$ 
& $w_1w_2^2w_3^3w_4^3$ 
& $w_1w_2^3w_3^3w_4^2$ 

\end{tabular}
\end{center}

\vspace{0.2cm}\noindent
Thus $\mathbb{P}({\boldsymbol \sigma}) 
= W({\bf y}_1) + W({\bf y}_2) + W({\bf y}_3) + W({\bf y}_4) $ indeed holds.
\end{example}

\section{Discussion}\label{sec:d}
In \cite{KMO3} the homogeneous $n$-TAZRP was studied and 
the steady state probabilities were obtained in 
the matrix product form $\mathrm{Tr}(X'_{\sigma_1}\cdots X'_{\sigma_L})$.
The operator $X'_\sigma$ has the structure of a corner transfer matrix \cite{Bax}
of a $\{{\bf 1}, {\bf a}^\pm, {\bf k}\}$-valued vertex model.
It is defined by \cite[eq.(5.10)]{KMO3} (with 
the interchange $\sigma^b\leftrightarrow \sigma^{n+1-b}$
to adjust the convention to this paper), and 
acts on the same space $F^{\otimes n(n-1)/2}$ as
the $X_\sigma$ (\ref{X}) in this paper does.
However the equality 
$X'_\sigma= \mathrm{const }\,X_\sigma |_{w_1=\cdots = w_n=1}$ 
up to permutations of the tensor components 
does not hold in general except $n=2$.
We hope to report on the relation 
between $X_\sigma$ and $X'_\sigma$, and thereby 
the approaches in this paper and \cite{KMO3} elsewhere.

\section*{Acknowledgments}
This work is supported by 
Grants-in-Aid for Scientific Research No.~15K04892,
No.~15K13429 and No.~23340007 from JSPS.

\end{document}